\documentclass[a4paper,onecolumn,11pt,accepted=2022-07-26]{quantumarticle}
\pdfoutput=1

\usepackage[utf8]{inputenc}
\usepackage[english]{babel}
\usepackage[T1]{fontenc}
\usepackage[numbers, sort&compress]{natbib}
\usepackage{amsmath,amsthm,amsfonts,amssymb}
\usepackage{mathtools}
\usepackage[pagebackref]{hyperref}
\usepackage{cleveref}
\usepackage{url}
\usepackage{braket}
\usepackage{physics}
\usepackage{thmtools,thm-restate}
\usepackage{graphicx}
\usepackage{xcolor}
\usepackage[boxruled,linesnumbered]{algorithm2e}
\usepackage{boxedminipage}
\usepackage{enumitem}
\usepackage{qcircuit}
\usepackage{tikz}
\usetikzlibrary{matrix}

\makeatletter
\def\@footnotecolor{red}
\define@key{Hyp}{footnotecolor}{%
 \HyColor@HyperrefColor{#1}\@footnotecolor%
}
\def\@footnotemark{%
    \leavevmode
    \ifhmode\edef\@x@sf{\the\spacefactor}\nobreak\fi
    \stepcounter{Hfootnote}%
    \global\let\Hy@saved@currentHref\@currentHref
    \hyper@makecurrent{Hfootnote}%
    \global\let\Hy@footnote@currentHref\@currentHref
    \global\let\@currentHref\Hy@saved@currentHref
    \hyper@linkstart{footnote}{\Hy@footnote@currentHref}%
    \@makefnmark
    \hyper@linkend
    \ifhmode\spacefactor\@x@sf\fi
    \relax
  }%
\makeatother

\hypersetup{
  linkcolor  = blue!80!black,
  citecolor  = green!70!black,
  footnotecolor = red!90!black,
  colorlinks = true,
}

\newcommand{\bitset}{\{0,1\}}
\newcommand{\fbitset}{\{1,-1\}}
\newcommand{\eps}{\varepsilon}
\renewcommand{\set}[1]{\left\{#1\right\}}
\newcommand{\U}{\mathcal{U}}
\newcommand{\NP}{\mathcal{NP}}
\newcommand{\Ptime}{\mathcal{P}}
\newcommand{\PT}{\mathcal{PT}}
\newcommand{\QPT}{\mathcal{QPT}}
\newcommand{\MAP}{\mathcal{MAP}}
\newcommand{\QMAP}{\mathcal{QMAP}}
\newcommand{\QCMAP}{\mathcal{QCMAP}}

\newcommand{\IPP}{\mathcal{IPP}}
\newcommand{\UPP}{\mathcal{UPP}}
\newcommand{\PP}{\mathcal{PP}}
\newcommand{\QMA}{\mathcal{QMA}}
\newcommand{\QCMA}{\mathcal{QCMA}}
\newcommand{\BQP}{\mathcal{BQP}}

\newcommand{\MA}{\mathcal{MA}}

\newcommand{\E}{\mathbb{E}}
\newcommand{\F}{\mathbb{F}}
\renewcommand{\H}{\mathcal{H}}
\newcommand{\K}{\mathcal{K}}

\newcommand{\N}{\mathbb{N}}
\renewcommand{\P}{\mathbb{P}}
\newcommand{\R}{\mathbb{R}}
\renewcommand{\U}{\mathcal{U}}
\newcommand{\V}{\mathcal{V}}

\newcommand{\disj}{\mathsf{DISJ}}
\newcommand{\thrdeg}{\mathsf{thrdeg}}
\newcommand{\sgn}{\mathsf{sgn}}

\DeclareMathOperator{\pos}{pos}
\DeclareMathOperator{\poly}{poly}
\DeclareMathOperator{\polylog}{polylog}

\newtheorem{definition}{Definition}
\newtheorem{inftheorem}{Theorem}
\newtheorem{theorem}{Theorem}
\numberwithin{theorem}{section}
\newtheorem{corollary}{Corollary}
\newtheorem{lemma}{Lemma}
\newtheorem{proposition}{Proposition}

\theoremstyle{definition}
\newtheorem{remark}{Remark}
\newtheorem{openquestion}{Open Question}

\begin{document}

\title{Quantum Proofs of Proximity}

\author{Marcel Dall'Agnol}
\email{msagnol@pm.me}

\author{Tom Gur}
\affiliation{Department of Computer Science, University of Warwick, UK}
\email{tom.gur@warwick.ac.uk}
\thanks{Tom Gur is supported by the UKRI Future Leaders Fellowship MR/S031545/1.}

\author{Subhayan Roy Moulik}
\affiliation{Department of Mathematics, UC Berkeley, USA \& Department of Computer Science, University of Oxford, UK}
\email{srm@math.berkeley.edu}
\thanks{SRM is supported by the  National Science Foundation under the QLCI program through grant number OMA-2016245 and the Clarendon Fund.}

\author{Justin Thaler}
\affiliation{Department of Computer Science, Georgetown University, USA}
\email{justin.thaler@georgetown.edu}
\thanks{Justin Thaler is supported in part by NSF CAREER award CCF-1845125,  NSF SPX award CCF-1918989, and DARPA under Agreement No. HR00112020022. Any opinions, findings and conclusions or recommendations expressed in this material are those of the author and do not necessarily reflect the views of the United States Government or DARPA.}





\maketitle
\begin{abstract}
We initiate the systematic study of QMA algorithms in the setting of property testing, to which we refer as \emph{QMA proofs of proximity} (QMAPs). These are quantum query algorithms that receive explicit access to a sublinear-size untrusted proof and are required to accept inputs having a property $\Pi$ and reject inputs that are $\eps$-far from $\Pi$, while only probing a minuscule portion of their input.

We investigate the complexity landscape of this model, showing that QMAPs can be \emph{exponentially} stronger than both classical proofs of proximity and quantum testers. To this end, we extend the methodology of Blais, Brody, and Matulef (Computational Complexity, 2012) to prove quantum property testing lower bounds via reductions from communication complexity. This also resolves a question raised in 2013 by Montanaro and de Wolf (cf.\ Theory of Computing, 2016).

Our algorithmic results include a purpose an algorithmic framework that enables quantum speedups for testing an expressive class of properties, namely, those that are succinctly \emph{decomposable}. A consequence of this framework is a QMA algorithm to verify the Parity of an $n$-bit string with $O(n^{2/3})$ queries and proof length.
We also propose a QMA algorithm for testing graph bipartitneness, a property that lies outside of this family, for which there is a quantum speedup.

\end{abstract}

\newpage
\setcounter{tocdepth}{3}
\tableofcontents
\newpage

\section{Introduction}

Quantum property testing is a fundamental model of sublinear-time quantum computation. Its importance stems both from the practical difficulty in manipulating large quantum states, as well as from the fertile ground that it provides for complexity theoretic investigations of the power of quantum mechanics as a computational resource.
 Accordingly, this model has garnered a large amount of attention in the last decade (see, e.g., \cite{CFMW10, HA11, ACL11, CM13, OW15, ABRW16, NV17, AA18, BOW19, GL20, BCL20}, and the survey \cite{MdW16}).

Building on the vast literature of classical property testing (cf.\ the recent book \cite{G17}), quantum testers are defined as quantum query algorithms that solve the \emph{approximate} decision problem of membership in a subset $\Pi$ (of possibly quantum objects), which is typically referred to as a \emph{property}; that is, the tester must accept if its input is in the property $\Pi$ and reject if it is \emph{far} from $\Pi$ with respect to a natural metric.

This paper is concerned with the notion of QMA proofs, the quantum analogue of NP proofs, in property testing. Namely, we investigate the following question:

\begin{center}
\emph{What is the power of QMA proofs for quantum property testing?}
\end{center}



\subsection{Quantum proofs of proximity}
\label{sec:qpop}

The question of decision versus verification is foundational in theoretical computer science, and extends far beyond $\mathcal{P}$ vs.\ $\mathcal{NP}$. Indeed, the study of \emph{classical} proof systems in the property testing setting is well established \cite{EKRR04, BGHSV06, RVW13, KR15, GG21, BCGRS17, GR17, BRV18, GR18, GLR21, GGR18, RRR21, RR20}. These objects are called \emph{proofs of proximity}, and include, among others, PCPs of proximity, interactive proofs of proximity, and MA proofs of proximity (MAPs). We henceforth adopt this standard terminology, noting it is synonymous with \emph{proof systems in the property testing setting}.

A \emph{Quantum Merlin-Arthur} (QMA) proof of proximity protocol for a property \emph{of unitaries} $\Pi$,\footnote{We note that testing unitaries is not the only sensible definition; see \Cref{rem:operators-vs-states} for a discussion.} with respect to proximity parameter $\eps$, is defined as follows. The \emph{verifier}, a computational device given oracle access to a unitary $U$, receives a quantum state $\ket{\phi}$ from an all-powerful but untrusted \emph{prover}. Making use of these two resources, it must decide whether $U \in \Pi$ or $U$ is $\eps$-far from $\Pi$ with respect to a specific metric.\footnote{Note that, unlike in the classical case, where Hamming distance is with few exceptions the natural choice, there are many natural metrics on the set of unitary matrices (e.g., those induced by the operator or Hilbert-Schmidt norms).} Such a protocol is said to verify (or \emph{test}) $\Pi$ if, with high probability, the verifier accepts in the former case and rejects in the latter (see \Cref{sec:qmap} for a formal definition). We remark that the notion of QMA proofs of proximity is implicit in the literature as QMA query algorithms for approximate decision problems (e.g., the permutation testing problem \cite{A12, ST19}). In this paper, we initiate the systematic study of the notion of QMA proofs of proximity (QMAPs), and explore its power and limitations.

The complexity of a QMAP protocol is measured with respect to the amount of resources required by the verifier. Namely, we will evaluate the efficiency of a protocol by its \emph{proof complexity} $p$ (the number of qubits in the proof $\ket{\psi}$) and \emph{query complexity} $q$ (the number of oracle calls made by the verifier for a worst-case input $U$). In particular, both parameters should be \emph{sublinear} in nontrivial protocols.

Before proceeding to state our results, we briefly discuss three applications that underscore the motivation to study quantum proofs of proximity.

\paragraph{Delegated remote quantum computation.}
 With efficient data structures such as QRAM \cite{GLM08} and reliable communication channels, one may envision a new paradigm of computing -- where the data is stored in a trusted data centre, and can be accessed (in coherent superposition) by a remote quantum computer (c.f. \cite{LHJ22}). This allows a $O(n)$-qubit computer to access $2^n$ different locations of the data in single query. 
This setting motivates the need for protocols for \emph{delegated remote computation}, where a remote client interacts with a (powerful) server in order to perform computation (on remote data) that it could not perform on its own, but without trusting the server. Current experimental developments such as \cite{Mothers14, Lothers19, Dothers21} further support this possibility \cite{WEH18}.

In this setting, suppose a client wants to compute a function $f$ on the data, $x$. While the client cannot even load the entire database $x$, the server may compute and send $y=f(x)$ to the client and append a proof of proximity that asserts $f(x) = y$. This enables the client to check, with sublinear resources, that $x$ is not far from satisfying $f(x)=y$, i.e., that $x$ is close in, say, Hamming distance to an input $x'$ satisfying $f(x')=y$. Quantum proofs may allow for efficient remote \emph{quantum} computation. This extends the study of verifiable quantum computing (c.f. \cite{GKK19}) to the setting where the input is held in a data centre.

\paragraph{Quantum certification.} Suppose a manufacturer produces a device that it claims implements a quantum circuit with high fidelity. However, we do not have access to the architecture or authority to crack open the device and examine it. The physical device is given to us as a black box, into which we may feed a quantum state and receive another quantum state as output.

If we do not trust the manufacturer and would like to assert that the device is indeed implementing the claimed functionality, one alternative is to perform tomography and characterise its input-output behaviour. This needs extensive resources, however; to characterise an $n$-bit transformation, $2^{\Omega(n)}$ uses of the device are required \cite{HHJWY17, OW16}. Alternatively, we could require the manufacturer to provide a QMA proof of proximity that certifies the operation of the device is at least \emph{close} to what is expected.

The task of benchmarking and certification of quantum devices has been a prominent topic of research \cite{Eothers20}. In fact, the idea of using tools from property testing to this end has been suggested in past works \cite{C00, HM13, HLM17, BOW19} and is closely related to self-testing \cite{TKTB18, RKB18, SB20} and the more recent programme, on quantum algorithmic measurements, that gives a complexity theoretic perspective on experimental physics and already has applications in characterisation of quantum many-body systems \cite{ACQ22}. 

\paragraph{Complexity class separations in property testing.}
A fundamental question in quantum complexity is to determine whether $\QMA$ (the quantum analog of $\NP$) is strictly more powerful than $\BQP$ (the quantum analog of $\Ptime$). Quantum proofs of proximity offer a natural setting to study analogous questions.
As we explain shortly, we show unconditionally that in the property testing setting, $\QMA$ protocols are exponentially more powerful
than both $\BQP$ and (classical) $\MA$ protocols for certain problems,
and thus that QMA proofs of proximity are ``larger than the sum of their parts'' (see \Cref{sec:algorithms}). We further show that they nevertheless have limited power, and investigate the complexity landscape surrounding QMAPs (see \Cref{sec:complexity-separations}).

\begin{remark}[Operators vs.\ states]
    \label{rem:operators-vs-states}
    The fact that the set of objects subject to testing are \emph{unitaries} is a definitional matter of importance. Property testing of quantum objects can be (roughly) divided into two classes: testing unitaries, as in this work, and testing \emph{states}. The complexity of quantum testers for the latter is generally quantified by the number of copies of the state that are necessary (see, e.g., \cite{MdW16}). Testing (pure) states can be seen as a special case of testing unitaries: a property $\Pi$ of $n$-qubit states defines a property $\Lambda$ of $n$-qubit unitaries $U$ such that (say) $U \ket{0^n} \in \Pi$ with the metric on unitaries induced by that on states. (One may also generalise testing mixed states similarly by considering quantum channels.)
    
    However, if $\Pi$ is a set of states and allowing for \emph{entanglement between the state being tested and that given as proof} yields a markedly different model; indeed, its study may shed light on the \emph{quantum PCP conjecture} \cite{AAV13}.
\end{remark}

\subsection{Our results}
\label{sec:results}

Our main results are divided into two parts: \Cref{sec:complexity-separations} charts fundamental  aspects of the complexity landscape surrounding quantum proofs of proximity, while \Cref{sec:algorithms} covers algorithmic results, where we show sufficient conditions for properties to admit efficient QMAP protocols.

We write $\QMAP(\eps, p, q)$ for the class of $\eps$-testable properties by a QMA proof of proximity protocol with proof length $p$ and query complexity $q$ (non-calligraphic acronyms to refer to algorithms and protocols, while calligraphic letters denote complexity classes).

\subsubsection{Complexity separations}
\label{sec:complexity-separations}

Our first collection of results aims to chart the complexity landscape of quantum proofs of proximity. Recall that $\QMAP(\eps,p,q)$ is the class of $\eps$-testable properties by a QMAP with proof complexity $p$ and query complexity $q$. The classes $\MAP$ (MA proofs of proximity) and $\IPP$ (interactive proofs of proximity) are defined analogously. $\PT(\eps,q)$ and $\QPT(\eps,q)$ are the properties admitting classical and quantum $\eps$-testers with query complexity $q$, respectively, and $\QCMAP(\eps,p,q)$ is the restriction of $\QMAP(\eps,p,q)$ where the proofs are classical bit strings. Formal definitions of all of these classes can be found in \Cref{sec:prelims} and \Cref{sec:qmap-def}.

We write complexity classes with the parameters omitted (e.g., $\QMAP$) to denote the corresponding class of properties such that for some proximity parameter $\eps \in (0, 1)$ that is a universal constant, there is a protocol with proof and query complexities bounded by $\polylog(n)$.

We begin by showing the existence of a property that admits efficient QMAPs, yet neither quantum property testers nor MAPs can efficiently test it.

\begin{inftheorem}
    \label{infthm:qmap-vs-qpt-map}
    There exists a property $\Pi$ such that, for any small enough constant $\eps > 0$,
    \begin{align*}
    	\Pi &\in \QMAP(\eps, \log n, O(1)) \text{ and}\\
    	\Pi &\notin \QPT(\eps, o(n^{0.49}) ) \cup \MAP(\eps, p, q)
    \end{align*}
    when $p \cdot q = o(n^{1/4})$. In particular,
    \begin{equation*}
    	\QMAP \not \subseteq \QPT \cup \MAP\;.
    \end{equation*}
\end{inftheorem}

\Cref{infthm:qmap-vs-qpt-map} is, in fact, implied by a stronger result. We show that, for certain properties, MAPs are stronger than quantum testers, i.e. $\MAP \not \subseteq \QPT$ (\Cref{thm:map-vs-qpt}); and, for others, quantum testers are stronger than MAPs, i.e. $\QPT \not \subseteq \MAP$ (\Cref{thm:qpt-vs-map}). Combining these results, we conclude that $\QCMAP \not \subseteq \QPT \cup \MAP$, i.e., even QMAPs with classical proofs suffice to obtain an exponential speedup over both MAPs and quantum testers.

Having shown the aforementioned separation, a natural question poses itself: are there cases in which a \emph{quantum proof} cannot be substituted for a classical one? We observe that a straightforward adaptation of a known result shows this is indeed the case, albeit \emph{only for subconstant proximity parameters}: the $\QMA$ vs.\ $\QCMA$ oracle separation of Aaronson and Kuperberg \cite{AK07} carries over to the property testing setting, implying $\QMAP(1/\sqrt{n}, \log n, 1) \not \subseteq \QCMAP(1/\sqrt{n}, p, q)$ if $\sqrt{p}q = o(\sqrt{n})$ (see \Cref{sec:qmap-vs-qcmap} for details).

We then shift to proving \emph{limitations} on the algorithmic power of QMAPs, showing that there exist explicit properties that are extremely difficult for such protocols.
First, we observe that known lower bounds on the complexity of $\QMA$ protocols for the \emph{Permutation Testing} problem
\cite{A12, ST19} yield an explicit property that does
not have any QMA protocol of polylogarithmic proof length and query complexity, and in fact establishes that
$\IPP \not \subseteq \QMAP$ (see \Cref{sec:ipp-vs-qmap} for details). We thus obtain the complexity landscape shown in \Cref{fig:complexity}.

Finally, we present an entire class of properties that cannot be solved by efficient QMAP protocols.
This extends and simplifies one of the main results of \cite{FGL14}, which obtains the same result for classical MAPs.

\begin{inftheorem}[\Cref{cor:kwise}, informally stated]
    \label{infthm:kwise}
    If a non-trivial property $\Pi$ is $k$-wise independent and $\eps = \Omega(1)$ is sufficiently small, then $\Pi \notin \QMAP(\eps, p, q)$ when $pq = o(k)$.
\end{inftheorem}

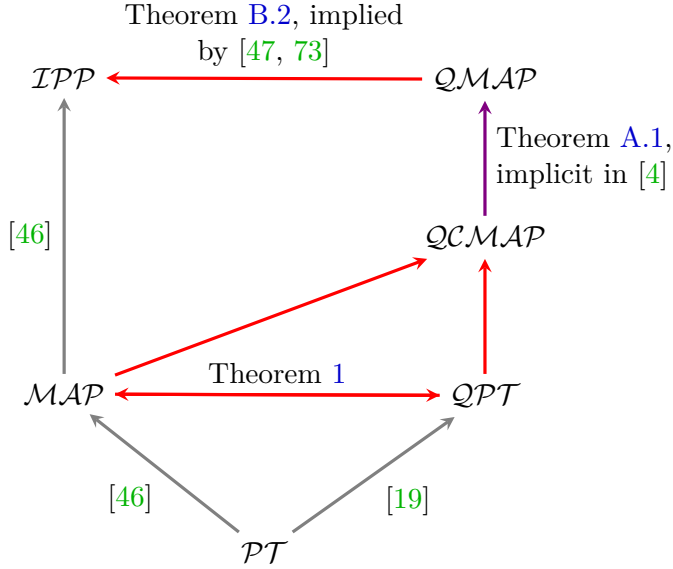
\begin{figure}
	\centering
	\begin{tikzpicture}
	  \matrix (m) [matrix of math nodes,row sep=4em,column sep=4em,minimum width=2em]
	  {
	  	\IPP & ~ & \QMAP\\
	  	~ & ~ & \QCMAP\\
		\MAP & ~ & \QPT \\
		~ & \PT & ~ \\};
	  \path[-stealth,very thick]
	  	(m-4-2) edge [color=gray] node [below left, color=black] {\cite{GR18}} (m-3-1)
	  			edge [color=gray] node [below right, color=black] {\cite{BFNR08}} (m-3-3)
	  	(m-3-1) edge [color=red] node [above, color=black] {\Cref{infthm:qmap-vs-qpt-map}} (m-3-3)
	  			edge [color=red] node [left] {} (m-2-3)
	  			edge [color=gray] node [left, color=black] {\cite{GR18}} (m-1-1)
	  	(m-3-3) edge [color=red] node {} (m-3-1)
	  			edge [color=red] node {} (m-2-3)
	  	(m-2-3) edge [color=violet] node [right, color=black, align=left] {\Cref{thm:qmap-vs-qcmap},\\ implicit in \cite{AK07}} (m-1-3)
	  	(m-1-3) edge [color=red] node [above, color=black, align=center] {\Cref{thm:ipp-vs-qmap}, implied\\by \cite{GLR21, ST19}} (m-1-1);
	\end{tikzpicture}
	\caption{Classification of complexity classes. An arrow from $\mathcal{A}$ to $\mathcal{B}$ is present when there exists a property requiring $n^{\Omega(1)}$ proof length or query complexity by algorithms of $\mathcal{A}$ but only $\polylog n$ proof/query complexity by algorithms of $\mathcal{B}$ (coloured red or grey when with respect to a proximity parameter $\eps = \Omega(1)$ that is a universal constant; and violet when $\eps \leq 1/\sqrt{n}$). The grey arrows are previously known separations.}
\label{fig:complexity}
\end{figure}

\subsubsection{Algorithms}
\label{sec:algorithms}

We show two general classes of properties whose structure allows for efficient QMAP (QMA proof of proximity) protocols. Moreover, these protocols only require \emph{classical} proofs (though the verifier is quantum).\footnote{Equivalently, all of the results in this section rely on QCMAP protocols. They are thus slightly stronger, continuing to hold if we replace $\QMA$ by $\QCMA$ in the theorem statements.} The first class is comprised of what we call \emph{decomposable properties}, which generalise the ``parametrised concatenation properties'' introduced in \cite{GR18}.

Roughly speaking, a property $\Pi$ is $(k,s)$-decomposable if testing whether $x \in \Pi$ can be reduced, with the help of a message of length $s$ from the prover, to that of testing whether $x^{(i)} \in \Lambda^{(i)}$ for $k$ smaller strings $x^{(i)}$ and properties $\Lambda^{(i)}$ (see \Cref{def:decomposable}). Since there may be several decompositions of the same string, the prover's message is said to \emph{specify} a decomposition, i.e., mappings $x \mapsto x^{(i)}$ and $\Pi \mapsto \Lambda^{(i)}$.

\begin{inftheorem}[\Cref{thm:decomposable}, informally stated]
\label{infthm:decomposable}
	If a property $\Pi$ is $(k,s)$-decomposable into strings of length $m$, each of which is $\eps$-testable by a MAP protocol with proof complexity $p$ and query complexity $q = q(m,\eps) = m^\alpha/\eps^\beta$, then
	\begin{equation*}
		\Pi \in \QMAP(\eps, s + kp, q')\;,
	\end{equation*}
	where $q'$ is much smaller than $q$ for many parameter values $\alpha$ and $\beta$.
\end{inftheorem}

The second class of properties amenable to QMA proofs of proximity are those admitting \emph{one-sided} (classical) MAPs which do not receive a proximity parameter explicitly, but rather reject strings $\eps$-far from the property with probability that is a function of $\eps$. Such algorithms are called \emph{proximity-oblivious} MAPs and readily admit quantum speedups (via the technique of amplitude amplification; see \Cref{sec:po-map}).

\begin{inftheorem}
\label{infthm:po-map}
	If a property $\Pi$ admits a proximity-oblivious MAP protocol with proof complexity $p$ and query complexity $q$, which always accepts $x \in \Pi$ and rejects when $x$ is $\eps$-far from $\Pi$ with probability $\rho(\eps) > 0$, then
	\begin{equation*}
		\Pi \in \QMAP\left(\eps, p, O\left(\frac{q}{\sqrt{\rho(\eps)}}\right)\right)\;.
	\end{equation*}
\end{inftheorem}

As applications of \Cref{infthm:decomposable}, we show:

\begin{corollary}
\label{cor:kmonotone}
	Let $\Pi_{k,[n]}$ denote the set of \emph{$k$-monotone} functions $f \colon [n] \to \bitset$, i.e., those which change from nondecreasing to nonincreasing and vice-versa at most $k-1$ times. For all $\eps \in (0,1)$,

	\begin{equation*}
		\Pi_{k,[n]} \in \QMAP\left(\eps, k \log n, \tilde{O}\left(\frac1{\eps}\right)\right)\;.
	\end{equation*}
\end{corollary}

\begin{corollary}[\Cref{cor:eulerian}, informally stated]
\label{infcor:eulerian}
  For every $k \in [n]$, the property of Eulerian graph orientations of $K_{2,n-2}$ is $\eps$-testable by a QMAP protocol with proof complexity $\tilde{O}(k)$ and query complexity $\tilde{O} \left(\frac{n}{k\sqrt{\eps}}\right)$.
\end{corollary}

Applying \Cref{infthm:po-map}, we also show the following:

\begin{corollary}[\Cref{cor:robp,cor:cfl}, informally stated]
\label{cor:robp-cfl}
	For every $k \in [n]$, acceptance by read-once branching programs and membership in context-free languages are both $\eps$-testable by QMAP protocols with proof complexity $\tilde{O}(k)$ and query complexity $O\left(\frac{n}{k \sqrt{\eps}}\right)$.
\end{corollary}

\Cref{infcor:eulerian,cor:robp-cfl} achieve a dependence on $\eps$ that is quadratically better as compared to the best known MAPs \cite{GR18, GGR18}, while \Cref{cor:kmonotone} is more efficient than the best known (classical) testers and MAPs for a wide range of parameters (see \Cref{sec:k-monotonicity}).

Classically, casting \emph{exact decision} problems in the framework of proofs of proximity (i.e., testing with respect to proximity parameter $\eps = 1/n$) is completely trivial except for degenerate cases, as most functions of sublinear query complexity are extremely simple.
 Rather surprisingly, this is \emph{not} the case quantumly, and indeed setting $\eps = 1/n$ in \Cref{cor:kmonotone,infcor:eulerian,cor:robp-cfl} yields sublinear algorithms for the corresponding exact decision problems. For \emph{layered} branching programs, we also prove the following, which improves on the parameters of \Cref{cor:robp-cfl} and lifts the read-once restriction:

\begin{inftheorem}
\label{thm:bp}
	There exists a QMA protocol for acceptance of $n$-bit strings by layered branching programs of width $w = w(n)$ and length $\ell = \ell(n)$ with query complexity $O(\ell^{2/3})$ and proof complexity $O(\ell^{2/3} \log w)$.
\end{inftheorem}
For details on decomposability and its implications, see \Cref{sec:decomposable}.
Finally, we prove that QMAP protocols are useful beyond proximity-oblivious and decomposable properties. The problem of testing bipartiteness of a graph does not fit either class, yet admits an efficient protocol nonetheless (see \Cref{sec:bipartite}).

\begin{inftheorem}[\Cref{thm:bipartite}, informally stated]
\label{infthm:bipartite}
	Bipartiteness for \emph{rapidly-mixing} $n$-vertex graphs (in the bounded-degree model) is $\eps$-testable by a QMAP protocol with proof complexity $\tilde{O}(\sqrt{n})$ and query complexity $\tilde{O}(n^{1/3}/\eps^{5/6})$.\footnote{We remark that this test assumes the rapidly-mixing promise for both yes- and no-instances.}
\end{inftheorem}

\subsection{Technical overview}

In this section, we discuss the high-level ideas of the techniques used in the proofs of the results stated in \Cref{sec:results}. Our discussion is divided into lower bounds and algorithmic techniques.

In \Cref{sec:seps}, we introduce some of the lower bound techniques that we use in charting the complexity landscape of quantum proofs of proximity. En route, we extend the framework of Blais, Brody and Matulef \cite{BBM12} to show lower bounds for \emph{quantum} property testers. To the best of our knowledge, this is the first quantum testing lower bound proved via a reduction from quantum communication complexity, an open question raised by Montanaro and de Wolf \cite[Question 4]{MdW16}. In addition, we show how to prove lower bounds on QMAP algorithms via an argument about the threshold degree of Boolean functions.

In \Cref{sec:algs} we show how to construct quantum proofs of proximity for properties that can be decomposed into sub-problems, and we prove that these QMAP protocols outperform both quantum testers as well as classical proof of proximity protocols. Moreover, we give an overview of an efficient QMAP protocol for a natural property of bounded-degree graphs, \emph{bipartiteness}, which does not fall into the decomposability paradigm.

\subsubsection{Lower bounds}
\label{sec:seps}

In this section, we highlight two techniques that we exploit to prove complexity separations and limitations on QMAPs: (1) proving quantum testing lower bounds via reductions from quantum communication complexity \cite{BBM12}, which we use to show a separation between MAPs and quantum testers; and (2) proving QMAP lower bounds by studying the threshold degree of Boolean functions.

\paragraph{Quantum testing lower bounds via reductions from communication complexity.} The methodology of \cite{BBM12} has proven very successful for showing \emph{classical} property testing lower bounds. However, extending this methodology to the quantum setting poses an inherent difficulty that we expand upon next. Following the exposition of \cite{MdW16}, we illustrate the methodology and the difficulty in the quantum setting by considering the problem of testing whether a function $f \colon \bitset^n \to \bitset$ is $k$-linear, i.e., a Fourier character of weight $k$.

 We can obtain query complexity lower bounds on testers via a reduction from the randomised \emph{communication complexity} problem of disjointness, as follows. Recall that, in the disjointness problem, Alice receives $x \in \bitset^n$ and Bob receives $y \in \bitset^n$ (for lower bound purposes, we may assume without loss of generality that both bit strings are promised to have Hamming weight $k/2$ for some known $k \in [n]$), and their goal is to decide whether or not there exists an index $i\in[k]$ such that $x_i=y_i=1$, while communicating a minimal number of bits.

Suppose that there exists a property tester for $k$-linearity with query complexity $q$. We will use this tester to construct a communication complexity protocol for disjointness (i.e., deciding if, for every $i \in [n]$, either $x_i = 0$ or $y_i = 0$) as follows. First, Alice and Bob use shared randomness and simulate the tester on the input $f$, interpreted as a function mapping $\bitset^n$ to $\bitset$ defined as $f(z) = \bigoplus_{i \in [n]} z_i \cdot (x_i \oplus y_i)$. To simulate a query $f(z)$, Alice computes $A(z) = \bigoplus_{i \in [n]} z_i \cdot x_i$ and sends it to Bob, while Bob computes $B(z) = \bigoplus_{i \in [n]} z_i \cdot y_i$ and sends it to Alice. Since $f(z) = A(z) \oplus B(z)$, each query to $f$ incurs 2 bits of communication. Moreover, if $x$ and $y$ are disjoint, then $f$ is $k$-linear; and if they are not disjoint, $f$ is $\ell$-linear for some $\ell < k$, and is in particular $1/2$-far from every $k$-linear function. Therefore, the simulated tester indeed solves the communication problem, so that the $\Omega(k)$ lower bound for the latter implies an $\Omega(k)$ lower bound for testing $k$-linearity.

An attempt to extend this to \emph{quantum} testers, however, reveals a severe bottleneck in the reduction. Note that, classically, the fact that Alice and Bob can use shared randomness to fix a \emph{deterministic} tester to simulate is crucial: at every step, both parties know which query $z$ the tester will make next \emph{without} the need to communicate it. The problem is that there is no way to fix the ``quantumness'' using shared randomness. Details follow.

While disjointness is still hard in the quantum communication complexity model, communicating the query (which may be in a superposition) that the quantum tester requires will incur a \emph{linear} overhead, rendering the reduction useless. Namely, to simulate a query to $f$ \emph{in superposition}, the parties need to exchange all $n$ qubits at each round: Alice would apply the  unitary (on $n+1$ qubits) $\ket{z}\ket{b} \mapsto \ket{z}\ket{b \oplus A(z)}$, and send the $(n+1)$-qubit state to Bob, who applies $\ket{z}\ket{b \oplus A(z)} \mapsto \ket{z}\ket{b \oplus A(z) \oplus B(z)} = \ket{z}\ket{b \oplus f(z)}$ and returns them to Alice. Simulating a single query then requires the communication of $2n+2$ qubits, rather than the $2$ needed by a classical tester. Thus, the reduction can only prove a degenerate $\Omega(1)$ testing lower bound. This is, in fact, not suprising, since $k$-linearity is testable with $O(1)$ queries by the Bernstein-Vazirani  \cite{BFNR08} algorithm!

While the discussion above might suggest that communication complexity can only yield trivial quantum testing lower bounds, we show this is not the case; indeed, the absence of nontrivial (quantum) applications of the technique thus far points to a conceptual barrier that is clarified by a coding-theoretic perspective similar to \cite{G20}, which reveals that the linear overhead is not inherent to any quantum reduction. Observe that testing $k$-linearity is a special case of \emph{testing a subset of a code}: namely, a $k$-linear function $f$ where $f(z) = w \cdot z$ corresponds to the \emph{Hadamard encoding} of the string $w$ with Hamming weight $k$ (which maps $w \in \bitset^n$ into the codeword $C(w) = (w \cdot z : z \in \bitset^n)$ with blocklength $n' = 2^n$). Note that the aforementioned quantum simulation strategy is efficient, in the sense that it requires only $O(\log n')$ qubits to communicate a representation of the length-$n'$ encoding; the issue is the Hadamard code's exponential blocklength $n' = 2^n$, which renders the simulation's efficiency moot. As we see next, however, \emph{the same reduction} yields nontrivial bounds if we choose the code appropriately.

Given a linear code $C \colon \bitset^n \to \bitset^{n'}$ with $n'  = \poly(n)$, only $\log n' = O(\log n)$ qubits are necessary to represent $C(x)$ and $C(y)$. More precisely, Alice can apply the $O(\log n)$-qubit unitary $\ket{i}\ket{b} \mapsto \ket{i}\ket{b \oplus C(x)_i}$ and send all $O(\log n)$ qubits to Bob, who applies $\ket{i}\ket{b} \mapsto \ket{i}\ket{b \oplus C(y)_i}$ and returns them. This composition of unitaries is
\begin{equation*}
	\ket{i}\ket{b} \mapsto \ket{i}\ket{b \oplus C(x)_i \oplus C(y)_i} = \ket{i}\ket{b \oplus C(x \oplus y)_i} \:,
\end{equation*}
simulating a query with \emph{logarithmic}, rather than linear, overhead.

Therefore, the $\Omega(\sqrt{n})$ quantum communication lower bound for disjointness \cite{R03} implies an $(n')^{\Omega(1)}$ lower bound for the problem of testing a subset of $C$. Indeed, we show that, for a linear code $C \colon \F^n \to \F^{n'}$ (over a larger field of odd characteristic), the property $\set{C(z) : z \in \bitset^n}$, of \emph{Booleanity},\footnote{In fact, we show (and use to prove the separation $\QPT \not \subseteq \MAP$) a lower bound for \emph{non}-Booleanity; but the symmetry of the model of communcation complexity implies the same lower bound holds for Booleanity as well.} which may be of interest in PCP constructions, has a quantum testing lower bound of $\Omega(\sqrt{n}/\log n)$ via a reduction from disjointness (see \Cref{sec:map-vs-qpt} for details). We remark that since this technique is used to show a separation between quantum testers and MA proofs of proximity, we use codes that are \emph{locally testable} and \emph{relaxed locally decodable}, which allow for efficient testing by a MAP. Since there exist such codes with a nearly-linear blocklength \cite{BGHSV06, AS21}, the lower bound we obtain is only slightly worse than a square root.

\paragraph{QMAP lower bounds via threshold degree.} We prove lower bounds for QMAPs via the \emph{threshold degree} of related functions.
A function $f \colon \bitset^n \to \bitset$ is said to have threshold degree (at most) $d$ if there exists a degree-$d$ polynomial $P(X_1, \ldots, X_n)$ over $\R$ such that $f(x) = 1$ if $P(x) > 0$ and $f(x) = 0$ if $P(x) < 0$; in other words, the threshold degree of $f$ is the smallest degree of a polynomial that \emph{sign-represents} $f$.

As a first step, we show that the inclusion $\QMA \subseteq \PP$ \cite{MW05} (in the \emph{polynomial-time} setting, implied by the technique known as Marriott-Watrous amplification) carries over to the property testing setting, implying $\QMAP \subseteq \UPP$.\footnote{$\UPP$ is the query model version of the class $\PP$ of unbounded-error randomised algorithms, where in particular the amount of randomness available to the algorithm is unbounded. Since PP algorithms run in polynomial time, they may access at most a polynomial number of random bits; this restriction does not hold for $\UPP$.} Next, we show that the query complexity of a UPP algorithm that computes $f$ is exactly the threshold degree of $f$ (this result is folklore, but we provide a proof for completeness). Since a property $\Pi$ induces the (partial) function $f_\Pi$ such that $f_\Pi(x) = 1$ when $x \in \Pi$ and $f_\Pi(x) = 0$ when $x$ is $\eps$-far from $\Pi$, the query complexity of a UPP algorithm that ``tests'' $\Pi$ (i.e., computes $f_\Pi$) is a lower bound on the product $pq$ of the proof and query complexities of any QMAP protocol for testing $\Pi$. Finally, we show that if $\Pi$ is \emph{$k$-wise independent} (i.e., looks perfectly random on any subset of $k$ coordinates) and not too large, the threshold degree of $f_\Pi$ is at least $k$, so that $pq = \Omega(k)$ (see \Cref{sec:kwise} for details).

In particular, any code with linear dual distance and small enough rate is an example of a hard property for QMAPs, requiring proof and query complexities that satisfy $pq = \Omega(n)$ for proximity parameter $\eps = \Omega(1)$.

\subsubsection{Algorithmic techniques}
\label{sec:algs}

As a warm-up, consider the \emph{exact decision} problem of verifying that an $n$-bit string $x$ has even parity. This is maximally hard for both IP algorithms and quantum query algorithms, requiring $\Omega(n)$ queries to the bit string, and thus asymptotically no better than trivially querying every coordinate. As we will see next, however, QMA algorithms can capitalise on having a proof \emph{and} quantum processing power to break the linear barrier.

We rely on the technique of \emph{amplitude amplification} \cite{BHMT02} to obtain such an algorithm with sublinear proof and query complexities. Loosely speaking, amplitude amplification takes a (randomised) decision algorithm that always accepts yes-inputs and rejects no-inputs with probability $\rho$, and produces an algorithm with rejection probability $2/3$ (for no-inputs) using the former algorithm only $O(1/\sqrt{\rho})$ times as a subroutine.

We can thus obtain a QMA (query) algorithm for the parity problem as follows. The proof string specifies the purported parities of each block of an equipartition of the input $x\in\bitset^n$ into $p$ blocks of length $n/p$. The verifier first checks that the proof string has even parity, rejecting immediately otherwise. Then, the verifier performs amplitude amplification on the following subroutine: sample $i \in [p]$ uniformly at random, read the entire block of $n/p$ bits and check that its parity coincides with that claimed by the proof; if so, accept, and reject otherwise.

Note that the aforementioned subroutine always accepts if $x$ has even parity and the proof corresponds to the parity of every block. On the other hand, if a string has odd parity and the proof has even parity, \emph{at least one bit of the proof} disagrees with the corresponding block, so that the subroutine rejects with probability at least $1/p$. Since we need only repeat $O(\sqrt{p})$ times, each of which queries $n/p$ bits, the query complexity of our algorithm is $q = O(n/\sqrt{p})$; in particular, if $p = n^{2/3}$ then $q = O(n^{2/3})$.\footnote{Subsequent to discovering the QMA protocol for parity, we learned that the same result was obtained in unpublished prior work of Alessandro Cosentino, Robin Kothari, and John Watrous.}

This is a special case of a more general phenomenon, which holds for all \emph{decomposable} properties (see \Cref{sec:exact} for a discussion of how exact decision follows as a special case). Since amplitude amplification can only be applied to one-sided algorithms (i.e., those that always accept a valid input), we restrict our attention to this type of algorithm hereafter.

\paragraph{Decomposable properties.} Roughly speaking, a property $\Pi$ is said $(k,s)$-decomposable if a ``specification'' of length $s$ efficiently reduces testing $\Pi$ to testing $k$ smaller properties $\Lambda^{(1)}, \ldots, \Lambda^{(k)}$. More precisely, $\Pi$ is $(k,s)$-decomposable if, 
\begin{enumerate}
\item there exists some $s$-bit string that specifies a set of $k$ properties $\Lambda^{(i)}$ as well as $k$ strings $x^{(i)} \in \bitset^{m_i}$ whose bits are determined by a small number of bits of $x$; and 
\item $\eps$-testing $x \in \bitset^n$ with respect to $\Pi$ reduces to testing $x^{(i)}$ with respect to $\Lambda^{(i)}$ in the following sense: when $x \in \Pi$ then $x^{(i)} \in \Lambda^{(i)}$ for all $i \in [k]$, whereas when $x$ is $\eps$-far from $\Pi$, then $x^{(i)}$ is $\eps_i$-far from $\Lambda^{(i)}$ for some $\eps_i$ satisfying $\E_i[\eps_i] = \Omega(\eps)$, where the expectation over $i$ means that $i$ is sampled with probability proportional to $m_i$. 
\end{enumerate}
If the specification is short (i.e., $s = O(k \log n)$), we say $\Pi$ is \emph{succinctly} $k$-decomposable (see \Cref{sec:decomposable} for details). Decomposable properties generalise the notion of \emph{parametrised $k$-concatenation properties} introduced in \cite{GR18}, which corresponds to the special case of a $(k,0)$-decomposition that is an equipartition of the input string.

Our simplest example of a decomposable problem is that of testing the set of $k$-monotone functions $f: [n] \to \bitset$, i.e., functions that change from non-decreasing to non-increasing and vice-versa at most $k - 1$ times. A natural decomposition of this property is to specify the set of at most $k- 1$ ``critical points'', which induce a set of at most $k$ subfunctions $f_i$ that are monotone and overlap with $f_{i-1}$ and $f_{i+1}$ at their endpoints; then, it suffices to test for (1-)monotonicity of each subfunction. More precisely, this property is $(k, (k-1)\log n)$-decomposable (thus succinctly $k$-decomposable), and given the (alleged) critical points $n_1 < n_2 < \cdots < n_{k-1}$, the subproperty $\Lambda^{(i)}$ for odd (resp. even) $i$ is the set of non-decreasing (resp. non-increasing) functions on $[m_i]$, where $m_i = n_i - n_{i-1} + 1$ (with $n_0 = 1$ and $n_k = n$). Note, moreover, that if $f$ is $\eps$-far from $k$-monotone, then its absolute distance from all functions specified by the critical points is at least $\eps n$; thus, denoting by $\eps_i$ the distance of $f_i$ to $\Lambda^{(i)}$, we have $\sum_i \eps_i m_i \geq \eps n$, implying $\E_i[\eps_i] = \Omega(\eps)$. We remark that decomposing other properties (e.g., branching programs, context-free languages, and Eulerian graph orientations) is much less straightforward and often allows for breaking the property into any desired number of sub-properties, which in turn admits proof length versus query complexity tradeoffs. See \Cref{sec:decomposable} for details.

Given a $(k,s)$-decomposable property that admits MAPs for the subproperties $\Lambda^{(i)}$, a natural protocol for $\Pi$ is to sample $i \in [k]$ uniformly at random and execute the verifier for $\Lambda^{(i)}$. Note that, if these MAPs have proof complexity $p$ and query complexity $q$, the protocol for $\Pi$ has proof length $s + kp$. Moreover, $\E_i[\eps_i] = \Omega(\eps)$ means that a randomly chosen $i \in [k]$ is (in expectation) at distance roughly $\eps$ from $\Lambda^{(i)}$, so it is reasonable to expect that $O(1/\eps)$ classical repetitions of the base protocol would ensure a rejection with high probability, and that a QMAP protocol can make do with only $O(1/\sqrt{\eps})$ repetitions using amplitude amplification.

The above outline glosses over the fact that we have no information on \emph{the distribution of errors} $(\eps_1, \dots, \eps_k)$.
For example,
it may be that most $\eps_i$ are of the same order of magnitude (in which case a random $i \in [k]$ is likely to point to a mildly corrupted $x^{(i)}$),
or it may be that a few $\eps_i$ are very large while all other $\eps_i$ are small or even zero (in which case $x^{(i)}$ is unlikely to be corrupted for a random $i \in [k]$, but when it is, the amount of corruption is large).
Fortunately, this issue can be addressed by the technique of \emph{precision sampling} \cite{L87}, incurring a merely logarithmic overhead. We thus obtain a QMAP protocol for $\eps$-testing $\Pi$ with proof complexity $s + kp$ and query complexity comparable to $q$ (and often smaller; see \Cref{thm:decomposable} for details).

\paragraph{Bipartiteness testing.} Consider the problem of testing whether a bounded-degree graph $G$ (given as an oracle to its adjacency list) is bipartite or far from any bipartite graph. (Note that this is not a decomposable property.) There exists a MAP protocol for a promise variant of this problem, where graphs are \emph{rapidly-mixing} \cite{GR18}. We will show that it is possible to combine quantum speedups obtained by amplitude amplification \emph{and} by replacing a classical subroutine with a more efficient quantum analogue.

Let us first consider the (classical) MAP verifier for bipartiteness, which receives a subset of vertices $S$ of size $k$, allegedly on the same side of a bipartition, as a proof. To test with respect to proximity parameter $\eps$, the verifier repeats the following procedure: sample a uniformly random vertex $v$, take roughly $n/(k\eps)$ short (lazy) random walks starting from $v$, recording whether the walk ended at a vertex in $S$ as well as the parity of the walk (i.e., the parity of the number of non-lazy steps). If two walks start from the same vertex $v$ and end in $S$ with different parities, then reject; otherwise, accept. Setting $m \coloneqq n/k$, the query complexity of (one iteration of) the verifier is $m/\eps$ (ignoring constants and polylogarithmic factors).

If the graph is bipartite and the proof $S$ is indeed on the same side of a bipartition, there cannot exist two paths from the same vertex into $S$ with different parities (as that would imply a path of odd length with both endpoints on the same side). Therefore, the verifier always accepts in this case. If the graph is $\eps$-far from bipartite, however, each iteration finds evidence to this effect with probability $\Omega(\eps)$. Thus, the classical verifier samples a new vertex roughly $1/\eps$ times, for a total query complexity of $m/\eps^2$.

Now, one immediate way to improve this algorithm is to perform amplitude amplification: the resulting algorithm repeats the procedure $1/\sqrt{\eps}$ times, improving the query complexity to $m/\eps^{3/2}$. A second (and less straightforward) strategy is to use the quantum \emph{collision-finding} algorithm \cite{A07} to reduce the number of random walks taken from each vertex to $(m/\eps)^{2/3}$, as in \cite{ACL11}.\footnote{Here and throughout, we use the term collision-finding to refer to Ambainis's algorithm that, for any $f$ with $1$-certificate complexity at most $2$, uses $\Theta(n^{2/3})$ queries and with constant probability outputs a $1$-certificate when run on any input $x \in f^{-1}(1)$.} This strategy reduces the required number of queries to $m^{2/3}/\eps^{5/3}$, improving the dependency on $m$ but achieving a worse one on $\eps$.

Of course, this begs the question: why not apply both optimisations? Indeed, we show how to tweak the classical MAP verifier in order to do so, and thus simultaneously obtain the speedups from each of them.\footnote{Amplitude amplification requires that the algorithm be \emph{invertible}, i.e., be given by a unitary $A$, as the technique repeatedly applies $A$ and $A^{-1}$. For this reason, it is often said to apply to quantum algorithms without intermediate measurements (as these make an algorithm non-invertible), which is not the case for collision-finding. However, the (standard) \emph{principle of deferred measurement} (see, e.g., \cite{NC16}) allows us to transform any quantum algorithm $A$ into a reversible $A'$ with the same query complexity, and apply amplitude amplification to the latter.}

More precisely, sample a uniformly random vertex $v$ and let $g_v$ denote the mapping $r \mapsto (a,b) \in \bitset^2$ obtained by executing a random walk starting from $v$ with $r$ as its inner randomness, where $a = 1$ if the walk stops at a vertex in $S$ and $b$ is the parity of the walk. The collision-finding algorithm is capable of finding a pair $r_0, r_1$ such that $g_v(r_i) = (1,i)$ for $i \in \{0, 1\}$, if such a pair exists. The query complexity of the collision-finding algorithm is the domain size to a $2/3$ power, and, since we take $m/\eps$ walks from $v$, the number of queries is
 $(m/\eps)^{2/3}$. Although such a collision is not guaranteed to exist for all starting vertices $v$, it is for a fraction of roughly $\eps$ of them. By applying amplitude amplification to the procedure described in this paragraph, we obtain a QMAP protocol for bipartiteness with proof length $O(k \log n)$ and query complexity $\tilde{O}((m/\eps)^{2/3} \cdot 1/\sqrt{\eps}) = \tilde{O}((n/k)^{2/3}/\eps^{5/6})$ (see \Cref{thm:bipartite} for details).

\subsection{Open problems}

This work begins the exploration of quantum proofs of proximity, leaving a host of uncharted research directions. We wish to highlight a small number of open problems, which we find to be of particular interest.

In the diagram of complexity class separations of \Cref{fig:complexity}, an evident shortcoming is the absence of $\QMAP \not \subseteq \QCMAP$ in the natural setting of parameters, i.e., with proximity $\eps = \Omega(1)$ rather than $\eps = \Theta(1/\sqrt{n})$.

\begin{openquestion}
    \label{qtn:qmap-vs-qcmap}
	What is the largest proximity parameter $\eps$ such that
	\begin{equation*}
	    \QMAP(\eps, \polylog(n), \polylog(n)) \not \subseteq \QCMAP(\eps, p, q)
	\end{equation*}
	with some proof and query complexities satisfying $pq = n^{\Omega(1)}$?
\end{openquestion}

Given our focus on quantum MA (i.e., \emph{non-interactive}) proofs of proximity, it is natural to ask what is achievable by allowing quantum property testers to interact with quantum provers, as opposed to static proofs.

\begin{openquestion}
	What is the power of quantum IP proofs of proximity (QIPPs)?
\end{openquestion}

More specifically, it is known that there exist classical interactive proof of proximity (IPP) protocols with $\tilde{O}(\sqrt{n})$ proof and query complexities for large classes of languages \cite{RVW13, RR20}. Moreover, these complexities are optimal (up to polylogarithmic factors) for \emph{classical} protocols, under reasonable cryptographic assumptions \cite{KR15}. Could quantum interactive proofs break the square-root barrier?

\begin{openquestion}
	Can QIPPs test logspace-uniform $\mathcal{NC}$ languages with $o(\sqrt{n})$ proof and query complexities?
\end{openquestion}

Finally, while we show a strong lower bound for QMAPs for $k$-wise independent properties, they do not rule out the existence of sublinear QMAP protocols. Could a stronger lower bound be shown?

\begin{openquestion}
	Do there exist maximally hard properties for QMAPs, requiring $\Omega(n)$ query complexity when the proof complexity is $p = cn$ for some $c = \Omega(1)$?
\end{openquestion}

We note that the question has an easy (negative) answer if we take $c = 1$: with a proof (allegedly) equal to the input string $x$, a QMAP based on Grover search can test with $O(1/\sqrt{\eps}) = O(\sqrt{n})$ queries. Moreover, \cite{RS04} provides a (positive) answer for proximity parameter $\eps = 1/n$: there exists a property $\Pi \subset \bitset^n$ for which deciding whether $x \in \Pi$ or $x \notin \Pi$ requires proof and query complexities satisfying $p + q = \Omega(n)$ (indeed, this is true of most bipartitions of $\bitset^n$). Then taking $c = \Omega(1)$ small enough implies $q = \Omega(n)$. Thus, similarly to \Cref{qtn:qmap-vs-qcmap}, while we have answers for subconstant $\eps$, the problem is open when $\eps = \Omega(1)$.

\subsection*{Organisation}

The rest of this paper is organised as follows. In \Cref{sec:prelims}, we discuss the preliminaries for the technical sections. QMA proofs of proximity are formally defined in \Cref{sec:qmap}, and we prove our complexity class separations in \Cref{sec:qmap-vs-qpt-map} (which \Cref{sec:qmap-vs-qcmap,sec:ipp-vs-qmap} complement with separations implied by known results). \Cref{sec:kwise} proves lower bounds for QMAPs and concludes the complexity-theoretic part of the paper. Proceeding to the algorithmic part, we show in \Cref{sec:po-map} that QMAP protocols enable speedups for proximity-oblivious MAPs. In \Cref{sec:decomposable}, we define decomposability and prove the bulk of our algorithmic results, including exact decision problems as a special case. Finally, in \Cref{sec:bipartite} we show a QMAP protocol for testing graph bipartiteness.

\section{Preliminaries}
\label{sec:prelims}

We begin with standard notation. For an integer $\ell\geq1$, we denote by $[\ell]$ the set $\{1,2, \ldots, \ell\}$. We use $\polylog(n)$ to denote an arbitrary polylogarithmic function, i.e., a  polynomial in the logarithm of $n$. For ease of notation, we also define $N \coloneqq 2^n$.

We use $\H, \K$ to denote arbitrary finite-dimensional Hilbert spaces and use indices to differentiate between distinct spaces. The set of linear operators mapping $\H$ to $\K$ is denoted by $\mathcal{L}(\H,\K)$; the shorthand $\mathcal{L}(\H)$ stands for $\mathcal{L}(\H,\H)$. The set of positive semidefinite operators on $\H$ having unit trace is denoted by $\pos(\H)$. The set $T(\H, \K)$ consists of the linear mappings from $\mathcal{L}(\H)$ to $\mathcal{L}(\K)$.
We say $T$ is a completely positive map (CP-Map) if $T \otimes I_{\mathcal{L}(\H)}$ is positive for all $\H$, where $I_{\K}$ denotes the identity operator on a Hilbert space $\K$ . Furthermore, $T$ is a completely positive trace preserving map (CPTP map) if $T$ is a trace preserving CP-Map, i.e., such that $\Tr(T(\rho))= \Tr(\rho)$ for all $\rho$.
$\U(\H)$ denotes the set of unitary operators on $\H$. For a unitary transformation $U \in \U$, its conjugate transpose is denoted $U^{\dagger}$.

A pure state is a unit vector in the Hilbert space $\H$, and represented by the Dirac notation, e.g., $\Ket{\psi}$. 
A mixed state is a distribution on pure states $\{ p_k, \Ket{\psi_k}\}$, represented as a \emph{density matrix} $\rho = \sum p_k \Ket{\psi_k}\Bra{\psi_k}$. We write $\dim(\H)$ to denote the dimension of the Hilbert space $\H$. Also, for brevity, we represent $\Ket{0}^{\otimes n}$ as $\Ket{\mathbf{0}}$, for an arbitrary $n\geq 1$, where $\Ket{0}=[1\ \  0]^T$.

\paragraph{Complexity classes.} We use the convention of writing complexity classes in calligraphic capitals and denoting algorithms or protocols by latin capitals, e.g., $\BQP$ is a complexity class whose problems are solvable by BQP algorithms and $\mathcal{IP}$ is a complexity class characterised by IP protocols.

\paragraph{Quantum query model.}
A quantum algorithm $V$ with query access to a bit string $x \in \bitset^n$ is specified by a sequence of unitary operations $V_0 \ldots V_q$, that do not depend on the input $x$. A query to $x$ is given by the unitary $U_x$ on $\log n + 1$ qubits such that

\begin{equation*}
	U_x \Ket{i}\ket{b} = \Ket{i}\ket{b \oplus x_i}\;.
\end{equation*}

We denote by $V^U$ the output of $V$ with oracle access to a unitary $U$. Similarly, $V^U(n)$ denotes the case where $V$ has access to an additional explicit input $n$.

The final state of an algorithm that makes $q$ queries to the oracle, before measurement, is given by
\begin{equation*}
	V_q (U_f \otimes I) V_{q-1} (U_f \otimes I) \ldots V_1 (U_f \otimes I) V_0 \Ket{\mathbf{0}}\;.
\end{equation*}
The overall Hilbert space $\H$ used by the algorithm is split into three subspaces $\H_\text{in} \otimes \H_\text{w} \otimes \H_\text{out}$, and we denote by $\ket{\mathbf{0}} = \ket{0}^{\otimes \log(\dim \H)}$ the initial state of the verifier's memory. The oracle acts on the space $\H_\text{in}$, the workspace $\H_\text{w}$ can have arbitrary size and $\H_\text{out}$ represents the single-qubit output of the algorithm. The final step of the algorithm is to measure the $\H_\text{out}$ register in the computational basis and return the outcome.

\paragraph{Distance measures.} Unless otherwise stated, we will assume the distance measure $d$ for \emph{classical} objects, such as strings of symbols in a finite field $\F$, as that induced by normalised \emph{Hamming weight}: for $x, y \in \F^n$, the (normalised) Hamming weight of $x$ is $\abs{x} \coloneqq \abs{\set{i \in [n] : x_i \neq 0}}/n$ and the distance between $x$ and $y$ is $d(x,y) = \abs{\set{i \in [n] : x_i \neq y_i}}/n$. We will often refer to distances between bit strings, which corresponds to the case where $\F = \F_2 = \bitset$.

For unitary matrices $U, V$, unless otherwise stated, we consider the distance measure to be that induced by the normalised \emph{Hilbert-Schmidt norm}: the norm of $U$ is $\norm{U} \coloneqq \sqrt{\frac1{n}\sum_{i = 1}^n \sigma_i(U)^2}$, where $\sigma_i(U)$ is the $i^\text{th}$ eigenvalue of $A$ (in some arbitrary order); and the distance between $U$ and $V$ is $d(U,V) = \norm{U - V}$.

We say two objects $X, Y$ are \emph{$\eps$-close} if $d(X,Y) \leq \eps$, and otherwise they are \emph{$\eps$-far}. We also say $X$ is $\eps$-close to a set of objects $\set{Y_i}$ if there exists $i$ such that $d(X,Y_i) \leq \eps$.

\paragraph{Property testing.} An \emph{interactive proof of proximity} (IPP) for $\Pi$ is a proof system that solves the approximate decision problem of membership in $\Pi$. The \emph{verifier} algorithm receives as input a proximity parameter $\eps$ and has oracle access to $x$. It queries $x$ in at most $q$ coordinates and interacts with an all-powerful but untrusted prover by exchanging $m$ messages, where the total number of communicated bits is $c$. The verifier must accept when $x\in \Pi$ and reject when $x$ is $\epsilon$-far from $\Pi$, with bounded probability of error; the outcome of such an interaction is denoted $\langle P(x), V^x\rangle$. Formally,

\begin{definition}
  An \emph{interactive proof of proximity} (IPP) for a property $\Pi = \cup_n \Pi_n$ is an interactive protocol with two parties: a
  (computationally unbounded) prover $P$ and a verifier
  $V$, which is a probabilistic algorithm. The parties send
  messages to each other in turns, with the first sent from prover to verifier, and at the end of the communication, the
  following two conditions are satisfied:

  \begin{enumerate}
  \item \emph{Completeness:} For every $\eps > 0$, $n \in \N$, and
    $x \in \Pi_n$, there exists a prover $P$ such that
    \begin{equation*}
      \P \left[\langle P(x), V^x\rangle(n, \eps) = 1 \right] \ge 2/3\;,
    \end{equation*}
    where the probability is over the coin tosses of $V$.

  \item \emph{Soundness:} For every $\eps > 0$, $n \in \N$,
    $x \in \bitset^n$ that is $\eps$-far from $\Pi_n$ and for
    every computationally unbounded (cheating) prover $P^*$
    it holds that
    \begin{equation*}
      \P \left[\langle P^*(x), V^x\rangle(n, \eps) = 1 \right] \le 1/3\;,
    \end{equation*}
    where the probability is over the coin tosses of $V$.
  \end{enumerate}

  The \emph{query complexity} $q$ of the protocol is the maximum number of queries the verifier makes to $x$ in its execution; the \emph{round complexity} $m$ is the number of messages exchanged between prover and verifier; and the \emph{communication complexity} $c$ is the total number of bits communicated by these messages.

  The set of properties $\Pi$ for which there exists an IPP protocol with proximity parameter $\eps$ with $m$ messages, communication complexity $c$ and query complexity $q$ is denoted $\IPP(\eps, c, q, m)$.
\end{definition}

When the completeness condition holds with probability 1, i.e., the verifier always accepts when $x \in \Pi$, we call the protocol \emph{one-sided}. Moreover, if the verifier does not receive $\eps$ explicitly, but rejects inputs that are $\eps$-far from $\Pi$ with \emph{detection probability} $\rho(n,\eps) > 0$, the procotol is said to be \emph{proximity-oblivious}.

A \emph{Merlin-Arthur proof of proximity} (MAP) for $\Pi$ is an IPP where the entire communication is a single message from the prover to the verifier (i.e., an IPP with round complexity 1); the class $\MAP(\eps,p,q)$ is thus defined as $\IPP(\eps, p,q,1)$. The formal definition of the \emph{quantum} generalisation of MAPs is given in \Cref{sec:qmap-def}.

A \emph{property tester} is an IPP with $r = 0$, i.e., where no communication occurs. In this case, the verifier is called a \emph{tester}, and we define $\PT(\eps,q) \coloneqq \IPP(\eps, 0, q, 0)$.

\paragraph{Branching programs.} A branching program on $n$ variables is a directed acyclic graph that has a unique source vertex $v_0$ with in-degree 0 and (possibly) multiple sink vertices with out-degree 0. Each sink vertex is labeled either with 0 (i.e., reject) or 1 (i.e., accept). Each non-sink vertex is labeled by an index $i \in [n]$ and has exactly 2 outgoing edges, which are labeled by 0 and 1. The output of the branching program $B$ on input $x \in \bitset^n$ , denoted $B(x)$, is the label of the sink vertex reached by taking a walk, starting at the source vertex $v_0$, such that at every vertex labeled by $i \in [n]$, the step taken is on the edge labeled by $x_i$.

A branching program is said to be \emph{read-once} (or ROBP for short) if, along every path from source to sink, every index $i \in [n]$ appears at most once. The size of a branching program $B$ is the number of vertices in its graph.

A branching program is \emph{layered} if its nodes can be partitioned into $V_0$, $V_1$, \ldots, $V_\ell$, where $V_0$ only contains the source node, $V_\ell$ are the sink nodes and every edge is between $V_{i-1}$ and $V_i$ for some $i \in [\ell]$. The \emph{length} of a layered branching program is its number of (nontrivial) layers $\ell$, and its \emph{width} is the maximum size of its layers, i.e., $\max_{i \in [\ell]} \set{\abs{V_i}}$.

\paragraph{Coding theory.} A \emph{code} $C \colon \F^k \to \F^n$ (where $\F$ is a finite field) is an injective mapping from \emph{messages} of length $k$ to codewords of \emph{blocklength} $n$. The \emph{rate} of the code $C$ is $k/n$ and its \emph{distance} is the minimum, over all distinct messages $x,y \in\F^k$, of $d(C(x),C(y))$. We shall sometimes slightly abuse notation and use $C$ to denote the set of all of its codewords $\set{C(x) : x\in\F^k} \subset \F^n$. If the mapping $C$ is linear, we say $C$ is a linear code.

\section{QMA Proofs of Proximity}
\label{sec:qmap}

This section provides a formal definition of QMAPs.

\subsection{Definition}
\label{sec:qmap-def}
A \emph{quantum Merlin-Arthur proof of proximity} (QMAP) for a property $\Pi = \bigcup \Pi_n$ is a proof system consisting of a quantum algorithm $V$, called a verifier, that is given as explicit input an integer $n\in \N$ and a proximity parameter $\epsilon>0$. It has oracle access to a unitary $U \in \V \subseteq \U(2^n)$ acting on $n$ qubits, which belongs to a \emph{universe} $\V$ with an associated distance measure.\footnote{Note that this definition generalises classical properties, where $\V = \set{U_x : x \in \bitset^n} \subset \U(2^{n+1})$, the unitary $U_x$ acts as $U_x \ket{i}\ket{b} = \ket{i}\ket{b \oplus x_i}$, and the distance between $U_x$ and $U_y$ is the Hamming distance between $x$ and $y$.} Furthermore, the verifier receives a $p$-qubit quantum state $\rho$ explicitly as a purported proof that $U \in \V$.

We allow the oracles in the rest of the work to be \textit{quantum oracles}, i.e., CP-Maps. However, using the Stinespring dilation \cite{S55, P02}, we view them as unitary operators again, in a larger Hilbert space, and the formalism generalises readily. This allows us to work with most general transformations allowable by quantum theory.

The verfier $V$ receives $n$ and $\eps$ as inputs, and outputs a sequence of unitary operators $V_0 \ldots V_q$ that satisfies the two following conditions.

\begin{enumerate}
    \item \emph{Completeness.} For every $n\in \N$  and $U \in \Pi_n$, there exists a $p$-qubit quantum state $\ket{\psi}$ such that,\footnote{While the proof can also be a mixed state, assuming it to be pure is without loss of generality; see \Cref{rem:pure-state}.} for every proximity parameter $\epsilon>0$,
    \begin{equation*}
        \P \left[V^U (n,\eps, \ket{\psi}) = 1\right] \geq 2/3\;.
    \end{equation*}
    Equivalently, \emph{some} $p$-qubit quantum state $\ket{\psi}$ satisfies
    \begin{equation*}
        \norm{(\Ket{1}\Bra{1}\otimes I) W (\ket{\psi} \otimes \Ket{\textbf{0}})}^2 \geq 2/3\;,
    \end{equation*}
    where $\Ket{\textbf{0}}$ is the initial state of the verifier and $W$ the unitary obtained by interspersing $q$ calls to the oracle $U$ between the $V_i$; that is, $W = V_q (U\otimes I) V_{q-1} \ldots (U\otimes I) V_{0}$.

    \item \emph{Soundness.} For every $n\in \N$, $\eps > 0$ and $p$-qubit quantum state $\ket{\psi}$, if $U \in \U(2^n)$ is $\eps$-far from $\Pi_n$, then
    \begin{equation*}
        \P \left[V^U (n,\eps, \ket{\psi}) = 1\right] \leq 1/3\;.
    \end{equation*}
    Equivalently, \emph{every} $p$-qubit quantum state $\ket{\psi}$ satisfies
    \begin{equation*}
        \norm{(\Ket{1}\Bra{1}\otimes I) W (\ket{\psi} \otimes \Ket{\textbf{0}})}^2 \leq 1/3\;,
    \end{equation*}
    where $W = V_q (U\otimes I) V_{q-1} \ldots (U\otimes I) V_{0}$.
\end{enumerate}

The \textit{query complexity} of a QMAP is number of times the verifier calls the oracle $U$. More precisely, the query complexity is $q = q(n,\eps)$ if, for every $n \in \N$, $\epsilon>0$ and $U \in \U(2^n)$, the verifier makes at most $q$ queries to the input. Its \textit{proof complexity} is $p = p(n,\eps)$ if, for every $n\in \N$ and $U \in \Pi_n$, there exists a $2^p$-dimensional quantum state $\ket{\psi}$ satisfying both of the above conditions.

The running time $t_V=t_V(n,\eps)$ of the verifier is the minimum depth of the unitaries $V_0 \ldots V_q$, composed of gates from a fixed constant-qubit gate set. The running time $t_P=t_P(n,\eps,U)$ of the prover is similarly the minimum depth of the circuit that prepares the proof state $\ket{\psi}$.

\begin{definition}[$\QMAP$ complexity class]
    Fix a universe set of unitary operators $\V$ and distance measure $d: \V \times \V \rightarrow [0,1]$.
    $\QMAP(\eps,p,q,t_V,t_P)$ is the class of properties $\Pi \subseteq \V$ that admit a verifier for proximity parameter $\eps$ with query complexity $q$ and proof complexity $p$ such that the verifier and prover runtimes are $t_V$ and $t_P$, respectively.

    The complexity class $\QCMAP(\eps,p,q,t_V,t_P)$ is defined as above, with the additional restriction that the proof be \emph{classical}, i.e., that the $p$-qubit quantum state given as proof is a computational basis state.
\end{definition}

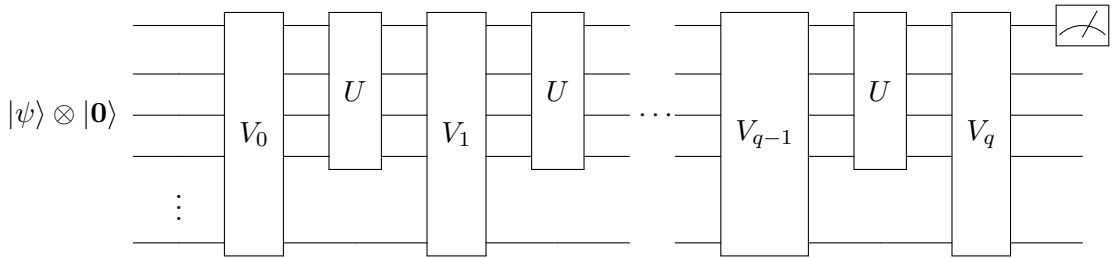
\begin{figure}[h]
  \raggedleft
  \begin{minipage}{0.9\textwidth}
	\[
		\Qcircuit @C=0.6cm @R=.2cm {
		  & \qw & \multigate{7}{V_0} & \multigate{3}{U} & \multigate{7}{V_1} & \multigate{3}{U} & \qw & & \multigate{7}{V_{q-1}} & \multigate{3}{U} & \multigate{7}{V_q} & \meter \\
		  & \qw & \ghost{V_0} & \ghost{U} & \ghost{V_1} & \ghost{U} & \qw & & \ghost{V_{q-1}} & \ghost{U} & \ghost{V_q} & \qw \\
		 \lstick{\ket{\psi} \otimes \ket{\mathbf{0}}} & \qw & \ghost{V_1} & \ghost{U} & \ghost{V_0} & \ghost{U} & \qw & \cdots ~~~~ & \ghost{V_{q-1}} & \ghost{U} & \ghost{V_q} & \qw \\
		  & \qw & \ghost{V_0} & \ghost{U} & \ghost{V_1} & \ghost{U} & \qw & &  \ghost{V_{q-1}} & \ghost{U} & \ghost{V_q} & \qw \\
		  &  & & & & & & & & & & \\
		  & \vdots & & & & & & & & & & \\
		  & & & & & & & & & & & \\
		  & \qw & \ghost{V_0} & \qw & \ghost{V_1} & \qw & \qw & & \ghost{V_{q-1}} & \qw & \ghost{V_q} & \qw
		}
	\]
  \end{minipage}
\caption{Schematic of a QMAP protocol that receives a proof state $\Ket{\psi}$ and makes $q$ queries to a unitary $U$.}
\end{figure}

For ease of notation, since the measures of complexity we will use throughout are \emph{proof} and \emph{query} complexity (but \emph{not} time complexity), we will often use $\QMAP(\eps,p,q)$ (and likewise for $\QCMAP)$ to denote the class as above, where $t_V, t_P$ are arbitrary functions.

We also denote the class of properties that are $\eps$-testable by quantum testers with $q$ queries as $\QPT(\eps, q)$, that is, $\QPT(\eps, q) \coloneqq \QMAP(\eps, 0, q)$.

\begin{remark}[Proofs are pure states]
    \label{rem:pure-state}
    Without loss of generality, the quantum state given as the proof is a pure state on $p$ qubits, i.e., a rank one positive semi-definite matrix. To see why, note that, if some mixed state $\rho = \sum p_i\Ket{\psi_i}\Bra{\psi_i}$ causes the verifier to accept with probability $2/3$, then, by convexity, there exists a state $\Ket{\psi_k}$ in that mixture that would also cause the verifier to output $1$ with probability at least $2/3$. Hence, the proof can be the pure state $\Ket{\psi_k}$. Likewise, if no pure state can make the verifier accept with probability larger than $1/3$, the same holds for mixed states.
\end{remark}


\section{Complexity separations}
\label{sec:qmap-vs-qpt-map}

Armed with a formal definition of QMAPs, we begin to chart the landscape of complexity classes to which quantum proofs of proximity belong. In \Cref{sec:complexity-prelim}, we provide definitions that will be necessary in the remainder of the section, mainly pertaining to coding theory.

Our main goal is to prove \Cref{infthm:qmap-vs-qpt-map}, namely, that QMAPs can exploit quantum resources and the availability of a proof to gain expressivity that neither can provide separately. This theorem follows from the \emph{incomparabilty} between the classes $\MAP$ and $\QPT$: in \Cref{sec:map-vs-qpt}, we exhibit a property $\Pi_B$ that is easy to test classically with a short proof, but requires many queries (without a proof) even for a quantum tester (\Cref{thm:map-vs-qpt}); moreover, in \Cref{sec:qpt-vs-map}, we show the existence of a property $\Pi_F$ that is easily testable quantumly but difficult to test classically, even with the aid of a proof (\Cref{thm:qpt-vs-map}).

The aforementioned results immediately imply the existence of a property, namely $\Pi_B \times \Pi_F$, which does not admit efficient MAPs nor quantum testers, requiring large proof or query complexity, whereas a QMAP with logarithmic proof and query complexities does exist (indeed, one with a \emph{classical} proof).

\begin{theorem}[\Cref{infthm:qmap-vs-qpt-map}, restated]
    \label{thm:qmap-vs-qpt-map}
    There exists a property $\Pi \subseteq \bitset^n$ such that, for any small enough constant $\eps > 0$,
    \begin{equation*}
    	\Pi \in \QCMAP(\eps, \log n, O(1))
    \end{equation*}
    and
    \begin{equation*}
    	\Pi \notin \QPT(\eps, o(n^{0.49}) ) \cup \MAP(\eps, p, q)
    \end{equation*}
    when $p \cdot q = o(n^{1/4})$.
\end{theorem}

\subsection{Preliminaries}
\label{sec:complexity-prelim}

We first define the necessary notions of local codes that will be used in this section. We denote throughout a finite field of constant size by $\F$.

\begin{definition}[Locally Testable Codes (LTCs)]
	A code $C \colon \F^k \to \F^n$ is locally testable, with respect to proximity parameter $\eps$ and error rate $\sigma$, if there exists a probabilistic algorithm $T$ that makes $q$ queries to a purported codeword $w$ such that:
	  \begin{enumerate}
  \item If $w = C(x)$ for some $x \in \F^k$, then $\P \left[T^w = 1 \right] \ge 1 - \sigma$.
  \item For every $w$ that is $\eps$-far from $C$, we have $\P \left[T^w = 0 \right] \ge 1 - \sigma$.
  \end{enumerate}
\end{definition}

Note that the algorithm $T$ that an LTC admits is simply an $\eps$-tester for the property of being a valid codeword of $C$.

\begin{definition}[Locally Decodable Codes (LDCs)]
  \label{def:ldc}
  A code $C \colon \F^k \to \F^n$ is locally decodable with \emph{decoding radius} $\delta$ and error rate $\sigma$ if there exists a probabilistic algorithm $D$ that given index $i \in[k]$ makes $q$ queries to a string $w$ promised to be $\delta$-close to a codeword $C(x)$, and satisfies
    \begin{equation*}
      \P[D^{w}(i) = x_i] \ge 1 - \sigma.
    \end{equation*}
\end{definition}

Since the best known constructions of LDCs have superpolynomial blocklength, we will make use of a relaxation of this type of code that allows for much more efficient constructions and suffices for our purposes.

\begin{definition}[Relaxed LDCs]\label{def:rldc}
  A code $C\colon\F^k \to \F^n$ with relative distance $\delta_C$ is a $q$-local relaxed LDC with success rate $\rho$ and decoding radius $\delta \in (0,\delta_C/2)$ if there exists a randomised algorithm $D$, known as a \emph{relaxed decoder} that, on input $i\in [k]$, makes at most $q$ queries to an oracle $w$ and satisfies the following conditions.

  \begin{enumerate}
  \item \textsf{Completeness}: For any $i\in [k]$ and $w = C(x)$, where $x\in\F^k$,
  \begin{equation*}
    \P[D^{w}(i) = x_i] \ge 2/3.
  \end{equation*}

  \item \textsf{Relaxed Decoding}: For any $i\in [k]$ and any $w \in\F^{n}$ that is $\delta$-close to a (unique) codeword $C(x)$,
    \begin{equation*}
      \P[D^{w}(i) \in \{x_i,\bot\}] \ge 2/3.
    \end{equation*}

    \item \textsf{Success Rate}: There exists a constant $\rho>0$ such that, for any $w\in\F^n$ that is $\delta$-close to a codeword $C(x)$, there exists a set $I_w\subseteq [k]$ of size at least $\rho k$ such that for every $i\in I_w$,
    \begin{equation*}
        \P[D^w(i)=x_i]\geq 2/3 \enspace.
    \end{equation*}
  \end{enumerate}
\end{definition}

As shown by \cite{BGHSV06, GGK19, CGS20, AS21}, there exist linear codes of only slightly superlinear blocklength that are both locally testable \emph{and} relaxed locally decodable:\footnote{We note that while LTCs and RLDCs are usually defined with respect to \emph{binary} alphabets, they can be constructed over larger fields as well as fields of odd characteristic. The constructions essentially rely on two components: a base (linear) code and a PCP of proximity, both of which readily extend to larger alphabets.}

\begin{theorem}
\label{thm:rldc}
    For any constant $\gamma > 0$, there exist linear codes over $\F$ with blocklength $n = k^{1+\gamma}$ that are relaxed locally decodable with $O(1)$ queries with respect to decoding radius $\delta = \Omega(1)$. Moreover, given any constant proximity parameter $\eps \in (0, \delta]$, the code is also locally testable with $O(1)$ queries.
\end{theorem}

Moreover, the tester and local decoder for these codes are one-sided (i.e., always accept when given a valid codeword as input), and the blocklength cannot be improved to linear \cite{GL21, DGL21}.

\paragraph{Communication complexity.} In the model of quantum communcation complexity, two parties with unbounded computational power aim to compute a joint predicate by communicating the smallest number of qubits with each other. Alice knows $x \in \bitset^k$, Bob knows $y \in \bitset^k$ and both hold a function $f: \bitset^{2k} \to \bitset$, and, by communicating qubits with each other, they must compute $f(x,y)$ with bounded probability of error. The \emph{communication complexity} of $f$ is the worst-case number qubits that need to be communicated in order to compute $f(x,y)$ over all $x, y$, minimised over all communication protocols.

We will make use of the well-known communication complexity problem of \emph{disjointness}.

\begin{definition}
\label{disj}
Let $x, y \in \bitset^k$ and $S, T \subseteq [k]$ be sets whose indicator vectors are $x$ and $y$, respectively; i.e., $S = \set{i \in [k] : x_i = 1}$ and $T = \set{i \in [k] : y_i = 1}$. Then $\disj_k(x,y) = 1$ if and only if $S$ and $T$ are disjoint, that is,
\begin{equation*}
    \disj_k(x,y) = \left\{\begin{array}{ll}1 & \text{ if } S \cap T = \varnothing\\0 & otherwise\end{array}\right. = \left\{\begin{array}{ll}1 & \text{ if } x_i = 0 \text{ or } y_i = 0 \text{ for all } i \in [k],\\0 & otherwise.\end{array}\right.
\end{equation*}
\end{definition}

This problem is known to be hard for quantum communication protocols, requiring $\Omega(\sqrt{k})$ qubits of communication, as shown in \cite{R03}.

\subsection{MAPs versus quantum testers}
\label{sec:map-vs-qpt}

We now set out to prove  \Cref{thm:map-vs-qpt}, which shows a property $\Pi_B$ that is efficiently testable with a short classical proof (\Cref{lem:Booleanity-in-qmap}) but for which a quantum tester must make a large number of queries (\Cref{lem:Booleanity-not-in-qpt}).

The property in question is defined as follows. First, consider a linear code $C:\F^k \rightarrow \F^n$ with $n= k^{1.001}$, over a field $\F$ of odd characteristic and size $O(1)$, that is both \emph{relaxed locally decodable} with $O(1)$ queries and decoding radius $\delta = \Omega(1)$; as well as \emph{locally testable} with $O(1)$ queries for any constant proximity parameter $\eps \in (0,\delta]$ (recall that \Cref{thm:rldc} shows that codes with these parameters exist).

The property $\Pi_B$ comprises the encoding of \emph{non-Boolean} messages, that is:
\begin{equation*}
    \Pi_B = \set{C(z) : z \in \F^k \setminus \bitset^k} \;.
\end{equation*}

We first show that the property $\Pi_B$ is efficiently testable via a MAP protocol with a short proof.
\begin{lemma}
\label{lem:Booleanity-in-qmap}
For any constant $\eps \in (0, \delta]$, $\Pi_B \in \MAP(\eps, \log n, O(1))$.
\end{lemma}

\begin{proof}
The verifier will follow the following strategy to test $\Pi_B$: first, test whether the input is $\eps$-close to the code $C$ (which can be accomplished with $O(1)$ queries due to the local testability of $C$). If the tester rejects, then not only is the input far from $\Pi_B$, but from all of $C$, in which case it rejects.

Except with small probability, if the tester accepts the input is $\delta$-close to $C$, so we may locally decode any coordinate of the message (also with $O(1)$ queries); using the proof string to determine this location, the verifier then checks if the symbol at that coordinate is Boolean-valued.

This strategy is laid out in Algorithm \ref{alg:Booleanity-map}.

\begin{center}
    \begin{algorithm}[H]
    \SetAlgoLined

    \KwIn{explicit access to a proximity parameter $\eps > 0$ and a proof string $\pi \in \bitset^{\log n}$, as well as oracle access to $x \in \F^n$.}

    Test if the input $w$ is a valid codeword with proximity parameter $\eps$. Reject if the test rejects.

    Interpret the proof as an index $i \in [k]$, locally decode $z_i$ and accept if $z_i \in \F \setminus \bitset$. Otherwise, reject.

    \caption{MAP verifier for $\Pi_B$}
    \label{alg:Booleanity-map}
    \end{algorithm}
\end{center}

Note that the proof complexity is $\log n$ by definition, and, since both local testing and local decoding have query complexity $O(1)$, the verifier makes $O(1)$ queries in total (note that querying an element of $\F$ requires $O(1)$ \emph{bit} queries, so the complexities of the tester and decoder are still constant in terms of bit queries). Moreover, if $w \in \Pi_B$, then $w = C(z)$ for some $z \in \F^k \setminus \bitset^k$, and local testing succeeds with probability 1; and when the prover specifies a coordinate $i$ such that $z_i \notin \bitset$, local decoding also succeeds with probability 1, so completeness follows.

Now, if $w$ is $\eps$-far from $\Pi_B$, then either (1) the input $w$ is $\eps$-far from \emph{any} codeword of the code $C$; or (2) $w$ is $\eps$-close to some $C(z)$ such that $z \in \bitset^k$.

In the first case, local testing (and thus the verifier) will reject with probability $2/3$. In the second case, the testing step may not trigger a rejection, but the local decoder then outputs, regardless of the proof $i \in [k]$, either $\bot$ or $z_i \in \bitset$ with probability $2/3$, any of which cause the verifier to reject.
\end{proof}

The next lemma shows that, unlike MAPs, quantum testers cannot test $\Pi_B$ efficiently.

\begin{lemma}
\label{lem:Booleanity-not-in-qpt}
Any quantum tester for the property $\Pi_B$ with constant proximity parameter $\eps \in [0,\delta)$ must have query complexity $\Omega(n^{0.49})$.
\end{lemma}

\begin{proof}

Recall that in the disjointness problem, Alice is given as input $x\in \bitset^k$, Bob is given $y\in \bitset^{k}$ and they must compute $\disj_k(x,y)$; and that $\Pi_B$ is the encoding of non-Boolean strings of length $k = n^{\frac{1}{1.001}}$ by the code $C$.

To show that any quantum tester needs $\Omega(n^{0.49})$ queries to $\eps$-test $\Pi_B$, we give a reduction showing that a tester with query complexity $q$ can be used to compute $\disj_k$ by communicating $O(q \log n)$ qubits. Since the quantum communication complexity of $\disj_{k}$ is $\Omega(\sqrt{k}) = \Omega(n^{\frac1{2} \cdot \frac1{1.001}}) = \Omega(n^{0.499})$, the query complexity of the quantum tester follows.

First, Alice and Bob use $C$ to encode $x$ and $y$, respectively. Now Alice holds $C(x) \in \F^n$ and Bob holds $C(y) \in \F^n$. Note that, defining $z \coloneqq x + y \in \F^k$, we have $\disj(x,y) = 0 \iff z \notin \bitset^k \iff C(z) \in \Pi_B$ (recall that the characteristic of $\F$ is larger than $2$, so that if $x_i = y_i = 1$, we have $z_i = 2 \notin \bitset$). Now, Alice and Bob respectively set up the unitaries $U_A$ and $U_B$ shown below, where the first register holds $\log n$ qubits and the second holds $O(1)$ qubits (enough to specify a single element of $\F$).

\begin{align*}
	\forall i \in [n], \alpha \in \F, \quad &U_A \ket{i}\ket{\alpha} = \ket{i}\ket{\alpha + C(x)_i}\\
	&U_B \ket{i}\ket{\alpha} = \ket{i}\ket{\alpha + C(y)_i}\;.
\end{align*}

Alice then simulates the quantum tester and only communicates with Bob in order to make a query to the oracle; if the tester accepts, Alice outputs 0, and she outputs 1 otherwise.

More precisely, whenever the tester calls the oracle $U$, which acts as $U \ket{i}\ket{\alpha} = \ket{i}\ket{\alpha + C(z)_i}$ on a $(\log n + O(1))$-bit quantum state $\rho$, Alice first applies $U_A$ to $\rho$, then sends all qubits to Bob; Bob then applies $U_B$ on the qubits it receives and returns them to Alice. This communicates a total of $O(\log n)$ qubits and implements the same transformation as querying $U$, since $U_B \cdot U_A \ket{i}\ket{j} =\ket{i}\ket{j + C(x)_i + C(y)_i} = \ket{i}\ket{j + C(x + y)_i}$ by the linearity of the code $C$ and the fact that $z = x + y$.

Each query made by the tester entails $O(\log n)$ qubits of communication, so that after $q$ queries, Alice and Bob exchange $O(q \cdot \log n)$ qubits in total. The tester accepts with probability at least $2/3$ when $C(z) \in \Pi_B \iff \disj(x,y) = 0$, in which case Alice outputs 0. If $\disj(x,y) = 1$, we have that $C(z)$ is $\delta$-far from $\Pi_B$ (since the relative distance of $C$ is $\delta$). Since $\eps \leq \delta$, the $\eps$-tester rejects with probability $2/3$ and Alice outputs 1 in this case.

Thus, Alice is able to compute $\disj_k$ with $O(q \cdot \log n)$ qubits of communication. Since the quantum communication complexity of $\disj_k$ is $\Omega(\sqrt k)$ \cite{R03}, we conclude that the tester must make $\Omega(\sqrt{k}/\log n) = \Omega(n^{0.499}/\log n) = \Omega(n^{0.49})$ queries.
\end{proof}

\begin{remark}
	Although we reduce $\disj$ to testing \emph{non}-Booleanity, a symmetric argument shows the same lower bound for the (arguably more natural) property of \emph{Booleanity} $\set{C(z) : z \in \bitset^k}$. Often, one key step in PCP constructions is to check that an encoding corresponds to a logical assignment, i.e., that it is the encoding of a Boolean message. Therefore, bounds on Booleanity may have consequences for PCPs.
\end{remark}

We conclude this section with the separation immediately implied by \Cref{lem:Booleanity-in-qmap} and \Cref{lem:Booleanity-not-in-qpt}.

\begin{theorem}
    \label{thm:map-vs-qpt}
    For every constant $\eps \in (0, \delta]$, the property $\Pi_B$ belongs to $\MAP(\eps, \log n, O(1))$ but does not belong to $\QPT(\eps, o(n^{0.49}))$. Therefore, \begin{equation*}
        \MAP(\eps, \log n, O(1)) \not \subseteq \QPT(\eps, o(n^{0.49}))\;.
    \end{equation*}
\end{theorem}

\subsection{Quantum testers versus MAPs}
\label{sec:qpt-vs-map}

In this section, we will show the existence of a property that is easily testable with a quantum tester, but for which a classical tester -- even with additional access to a proof -- must make a number of queries that depends strongly on the length of the input. More formally, we will show in \Cref{thm:qpt-vs-map} the existence of a property of $n$-bit strings in $\QPT(\eps, O(1/\eps))$ that is not in $\MAP(\eps,p,q)$ when $p \cdot q = o(n^{1/4})$ and $\eps$ is a small enough constant.

The property in question is derived from Forrelation, a problem that strongly separates classical and quantum algorithms in the query model; in fact, the work that proved such a separation already shows that it carries over to the property testing setting \cite{AA18}, which we will extend to the setting of MAPs. Formally, we have
\begin{lemma}[\cite{AA18}]
    \label{PiF_in_QPT}
    Define the property $\Pi_F$ as
    \begin{equation*}
        \Pi_F = \set{(f,g) : \Phi_{f,g} \leq 1/100},
    \end{equation*}
    where $(f,g)$ are $n/2$-bit strings corresponding to pairs of $\log (n/2)$-bit Boolean functions and $\Phi_{f,g} = (n/2)^{-3/2} \sum_{x, y \in \bitset^{\log(n/2)}} f(x) (-1)^{x \cdot y} g(y)$.

    \noindent Then, for any $\eps > 0$ sufficiently small,
    \begin{equation*}
        \Pi_F \in \QPT(\eps, O(1/\eps)) \text{ and } \Pi_F \notin \PT(\eps, o(\sqrt{n}/\log n)).
    \end{equation*}
\end{lemma}

Therefore, the property $\Pi_F$ is easy for quantum testers and hard for their classical counterparts. This section is thus devoted to showing that testing $\Pi_F$ is hard not only for property testers, but for MAPs as well: we will prove that a MAP for $\Pi_F$ requires proof length $p$ and query complexity $q$ satisfying $p q = \Omega(n^{1/4})$. We first introduce relevant definitions and theorems, then describe the steps of the proof.

Recall that $\MA$ is the class of languages that are decidable \emph{in polynomial time} with a (polynomial-size) proof string, the analogue of which is $\MAP$ in the property-testing setting. By \cite{HHT97}, $\MA$ is contained in the class $\mathcal{BPP}_\text{path}$ of languages decidable (with high probability) by a randomised Turing machine whose computational paths are all equally likely.\footnote{In $\mathcal{BPP}$, the probability of following a computational path is a function of its length (which coincides with the number of random coins flipped by the algorithm). $\mathcal{BPP}_\text{path}$ differs from $\mathcal{BPP}$ by lifting this restriction.} Query lower bounds for Forrelation (i.e., deciding whether $\abs{\Phi_{f,g}} \leq 1/100$ or $\Phi_{f,g} > 3/5$ for a pair $(f,g)$ of Boolean functions) against the latter are known:

\begin{proposition}[\cite{A10,C16}]
    \label{thm:bpp-path-lb}
    Any $BPP_\text{path}$ algorithm for Forrelation must make $\Omega(n^{1/4})$ queries to its input.
\end{proposition}

  We are now ready to describe the three steps taken in proving hardness of $\Pi_F$ for $\MAP$: we (1) show that transforming an MA algorithm with proof complexity $p$ and query complexity $q$ into a $BPP_\text{path}$ one \cite{HHT97} yields an algorithm with query complexity $O(pq)$; (2) show how a MAP for $\Pi_F$ implies an $\MA$ algorithm with the same parameters \emph{for Forrelation}; and (3) conclude that $pq = o(n^{1/4})$ implies a $\mathcal{BPP}_\text{path}$ upper bound of $O(pq) = o(n^{1/4})$ for the query complexity of Forrelation, which contradicts \Cref{thm:bpp-path-lb}.

The original proof of $\MA \subseteq \mathcal{BPP}_\text{path}$ (\cite{HHT97}, Theorem 3.7) takes an MA algorithm with proof complexity $p$ and constant success probability, repeats the execution $O(p)$ times, amplifying the success probability to $1 - O(2^{-p})$, and then defines a $\text{BPP}_\text{path}$ machine as follows. The machine (non-deterministically) guesses a proof string and simulates the MA algorithm with it, spawning ``dummy'' execution paths if the MA algorithm accepts. Inspecting this transformation in the query model, we obtain a quadratic overhead: if the MA algorithm has query complexity $q$ and proof complexity $p$, the $\text{BPP}_\text{path}$ machine thus obtained has query complexity $O(p q)$ (the $O(p)$ repetitions of the MA algorithm increase its query complexity multiplicatively by this amount, while the dummy paths make no queries).

The second step is formalised by the following lemma.
\begin{lemma}
    \label{lem:mapma}
    A MAP protocol for $\Pi_F$ with sufficiently small proximity parameter $\eps > 0$ implies an $\MA$ algorithm for Forrelation with the same query and proof complexities as the MAP protocol.
\end{lemma}
\begin{proof}
    We define an MA protocol \emph{for Forrelation} (as a gap problem) the natural way: the proofs and queries correspond to the proofs and queries of the MAP, and the MA verifier accepts if and only if the MAP rejects.

    To show correctness of this protocol, we follow the reduction of \cite{AA18}. Specifically, \cite[Lemma 40]{AA18} shows that any $(f',g')$ such that $f'$ is $\eps$-close to $f$ and $g'$ is $\eps$-close to $g$ satisfies $\abs{\langle f', H g' \rangle - \langle f, H g \rangle} = O(\sqrt{\eps} \log(1/\eps))$. By choosing a suitably small $\eps$, the right-hand side is at most (say) $1/100$. Thus any $(f,g)$ such that $\langle f, H g \rangle > 3/5$ is $\eps$-far from $\Pi$, and it follows that the MAP protocol will accept (with high probability) pairs $(f,g)$ such that $\abs{\langle f, H g \rangle} \leq 1/100$ and will reject if $\langle f, H g \rangle > 3/5$. Therefore, the MA protocol is able to distinguish between the two cases.
\end{proof}

A simple argument now proves the separation.

\begin{theorem}
    \label{thm:qpt-vs-map}
    Let $\eps > 0$ be a small enough constant. There exists a property $\Pi_F$ such that $\Pi_F \in \QPT(\eps, O(1/\eps))$ and $\Pi_F \notin \MAP(\eps,p,q)$ for any $p,q$ such that $p \cdot q = o(n^{1/4})$.
\end{theorem}
\begin{proof}
    Suppose, towards contradiction, that there existed a MAP for $\Pi_F$ with any proximity parameter $\eps$, as well as proof complexity $p$ and query complexity $q$ satisfying $p \cdot q = o(n^{1/4})$.

    By \Cref{lem:mapma}, there exists an MA protocol for Forrelation with the same query and proof complexities, which can then be transformed into a $BPP_\text{path}$ algorithm with query complexity $O(p \cdot q) = o(n^{1/4})$ for the same problem. But this contradicts the $\Omega(n^{1/4})$ lower bound of \Cref{thm:bpp-path-lb}.
\end{proof}

We are finally ready to prove the main separation, by exhibiting a property $\Pi$ in $\QCMAP$ which is in neither $\MAP$ nor $\QPT$.

\begin{proof}[Proof of \Cref{thm:qmap-vs-qpt-map}]
Recall that our goal is to show that the property $\Pi = \Pi_B \times \Pi_F$ that is efficiently $\eps$-testable by a QCMAP, but not by a quantum tester (without a proof) nor classicaly with a proof, for some small enough $\eps = \Omega(1)$. To this end, we invoke \Cref{thm:map-vs-qpt} and \Cref{thm:qpt-vs-map} and give the following verifier strategy explicitly. Note that, while the strings in $\Pi_B$ are over an alphabet $\F$ larger than $\bitset$, since $\abs{\F} = O(1)$ each symbol can be represented by $O(1)$ bits (and the proof indicates the first bit in such a block).

\begin{center}
    \begin{algorithm}[H]
    \SetAlgoLined

    \KwIn{explicit access to a proximity parameter $\eps' = 2\eps > 0$ and a proof $i \in [n/2]$, as well as oracle access to a string $z \in \bitset^n$}

    Interpret the input as a concatenation of $n/2$-bit strings $x$ and $y$. Use the algorithm of \Cref{lem:Booleanity-in-qmap} to verify, with $O(1)$ queries and the proof $i$, whether $x \in \Pi_B$ with proximity $\eps \coloneqq \eps'/2$.

    Use the quantum tester of \Cref{PiF_in_QPT} to test, with $O(1)$ queries, if $y \in \Pi_F$ with proximity $\eps$.

    If both of the previous tests accepted, then accept; otherwise, reject.

    \caption{QCMAP verifier for $\Pi_B \times \Pi_F$ }
    \end{algorithm}
\end{center}

Completeness follows immediately from \Cref{lem:Booleanity-in-qmap} and \Cref{PiF_in_QPT}, since the verifier for $\Pi_B$ accepts with certainty and the and tester for $\Pi_F$ accepts with probability $2/3$ when $x \in \Pi_B$ and $y \in \Pi_F$. If, on the other hand, $(x,y)$ is $\eps$-far from $\Pi_B \times \Pi_F$, then either $x$ is $\eps$-far from $\Pi_B$ or $y$ is $\eps$-far from $\Pi_F$, and either the verifier for $\Pi_B$ or the tester for $\Pi_F$ will reject (with probability $2/3$). Thus, the QCMAP verifier for $\Pi$ (executed with respect to proximity parameter $\eps' = 2\eps$) implies $\Pi \in \QCMAP(\eps, \log n, O(1))$.

All that remains is to show $\Pi = \Pi_B \times \Pi_F$ does not admit an efficient quantum tester nor a MAP. Assume, towards contradiction, that either $\Pi \in \QPT(\eps, o(n^{0.49}))$ or $\Pi \in \MAP(\eps, p, q)$ when $p \cdot q = o(n^{1/4})$. In the first case, applying the tester for $\Pi$ to  $\Pi_B \times \set{y}$ for some fixed $y \in \Pi_F$ shows that $\Pi_B \in \QPT(\eps, o(n^{0.49}))$, a contradiction with \Cref{lem:Booleanity-not-in-qpt}. In the second case, applying the MAP protocol for $\Pi$ to $\set{x} \times \Pi_F$, for some fixed $x \in \Pi_B$, shows that $\Pi_F \in \MAP(\eps, p, q)$, a contradiction with \Cref{thm:qpt-vs-map}.
\end{proof}


\section{A hard class of problems for QMAPs}
\label{sec:kwise}

When introducing a new complexity class in the landscape of known classes, it is important not only to exhibit problems it can solve, but also problems it cannot. We set out to show a natural limitation on QMAPs in this section, by answering (negatively) the following question: if a property ``looks random'' on any subset of $q$ coordinates, can a quantum proof be of any help to a verifier with query complexity $q$?
Intuitively, the answer should be no: 
if querying $q$ coordinates provides no information as to whether or not an input satisfies a property, then any proof (quantum or otherwise) should not be able to offer more information in conjunction with the queries than it does on its own.

We formalise this intuition in \Cref{thm:kwise}, which states the following: if a property $\Pi \subset \bitset^n$ is $k$-wise independent and sparse (i.e., its size $|\Pi|$ is sufficiently small compared to the set of all $2^n$ bit strings), then $k$
is a lower bound on the number of queries made by any randomised query algorithm that accepts all inputs in $\Pi$ with probability strictly greater than $1/2$, and rejects with probability strictly greater than $1/2$ when run on any input that is far from $\Pi$. In other words, the $\UPP$ query complexity of testing $\Pi$ is at least $k$ (recall that $\UPP$ is the query model version of $\PP$, which captures randomized computation with small bias). Note that \emph{some} assumption on the sparsity of $\Pi$ is necessary for any non-trivial lower bound to hold, if only to rule out, e.g., the trivially testable property $\Pi = \bitset^n$.

Combining \Cref{thm:kwise} with the well-known inclusion  $\QMA \subseteq \PP$ \cite{MW05} 
allows us to conclude the following:
for any $k$-wise independent and sufficiently sparse property $\Pi$, the product of proof and query complexities of a QMAP for verifying membership in $\Pi$ with constant proximity parameter $\eps$ is $\Omega(k)$ (see \Cref{cor:kwise}).

The proof of \Cref{thm:kwise} works as follows. Our analysis shows that the sparsity of $\Pi$ ensures there exists a subset $\Pi' \subset \bitset^n$ that is far from $\Pi$ such that $\Pi'$  is \emph{also} $k$-wise independent (see \Cref{lem:kwisecode}). This means that any query algorithm making fewer than $k$ queries cannot
distinguish a random input in $\Pi$ from a random input in $\Pi'$, as both sets ``look random'' when inspecting
only $k$ bits of a randomly chosen input from the set. Yet since $\Pi'$ is far from $\Pi$, any testing procedure for $\Pi$ must distinguish $\Pi$ from $\Pi'$. Hence, any tester for $\Pi$ must make $k$ queries (even if it only outputs the correct answer on inputs
in $\Pi$ and $\Pi'$ with probability strictly greater than $1/2$).

\paragraph{Technical Details.}
We begin recalling the definition of $k$-wise independence.
\begin{definition}
    A set of strings $S \subseteq \bitset^n$ is called $k$-wise independent if, for any fixed set of indices $I \subset [n]$ of size $k$, the string $x_{|I}$ is uniformly random when $x$ is sampled uniformly from $S$. Equivalently, for every $y \in \bitset^k$,
    \begin{equation*}
        \abs{\set{x \in S : x_{|I} = y}} = \frac{\abs{S}}{2^k}.
    \end{equation*}
\end{definition}

We next show that, given any small enough set $S$ of strings, there exists a $\Omega(n)$-wise independent set that is $\eps$-far from $S$. In the following lemma, $H$ denotes the binary entropy function $H(\alpha) = -\alpha \log \alpha - (1-\alpha) \log(1-\alpha)$ (whose restriction to $[0,1/2]$ is bijective).
\begin{lemma}
    \label{lem:kwisecode}
    Let $\eps \in (0,H^{-1}(1/5))$ and $S \subseteq \bitset^n$ be such that $\abs{S} < 2^{(1/4 - H(\eps)) n}$. Then there exists a linear code $C$ that is $\eps$-far from $S$ with dual distance $\Omega(n)$; equivalently, $C$ is $\Omega(n)$-wise independent.
\end{lemma}
\begin{proof}
    Let $C \colon \bitset^{3n/4} \to \bitset^{n}$ be a random linear code (where each entry of its generator matrix is a Bernoulli($1/2$) random variable). Then, for every $x \in \bitset^{3n/4}$, the codeword $C(x)$ is uniformly random in $\bitset^n$ (but \emph{not} independent of other codewords). Denoting by $N_\eps(S)$ the $\eps$-neighbourhood of $S$ (i.e., the set of bit strings at distance at most $\eps$ from $S$), we have:
    \begin{align*}
        \P[C \cap N_\eps(S) \neq \varnothing] &\leq \sum_{x \in \bitset^{3n/4}} \P[C(x) \in N_\eps(S)] = 2^{3n/4} \cdot \frac{\abs{N_\eps(S)}}{2^n}\\
        &\leq \frac{\abs{S} \cdot 2^{H(\eps) n}}{2^{n/4}} = o(1),
    \end{align*}
    and, by the probabilistic method, there exists a code $C \subset \bitset^n$ of size $2^{3n/4}$ that is $\eps$-far from $S$. Moreover, the dual code $C^\perp \colon \bitset^{n/4} \to \bitset^n$ is a linear code whose distance meets the Gilbert-Varshamov bound with high probability; that is, the distance of this dual code is $\Omega(n)$ with probability $1 - o(1)$, proving the claim.
\end{proof}

The previous lemma, when applied to a ``random-looking'' set $S$ (i.e., a $k$-wise independent $S$, for $k = o(n)$), will ensure that $S$ and the code $C$ are hard to distinguish. To make this precise, we first recall the definition of the \emph{threshold degree} of a (partial) function.

\begin{definition}
    Let $\mathcal{X} \subseteq \fbitset^n$ and let $f: \mathcal{X} \to \fbitset$ be any function defined on domain $\mathcal{X} \subseteq \fbitset^n$.\footnote{For notational convenience, we consider Boolean functions with codomain $\fbitset$, noting that this is equivalent to the usual codomain $\bitset$ by mapping $0 \to 1$, $1 \to -1$, and $\oplus$ to multiplication.} The threshold degree of $f$, denoted $\thrdeg(f)$, is the minimal degree of an $n$-variate polynomial $p$ that sign-represents $f$, i.e., such that $f(x) = \sgn(p(x))$ for all $x \in \mathcal{X}$.\footnote{Here, $\sgn(t)$ is defined to equal $1$ if $t > 0$, $-1$ if $t < 0$, and $0$ if $t=0$.} Note that no constraints are placed on the behaviour of $p(x)$ at inputs in $\fbitset^n \setminus \mathcal{X}$.
\end{definition}

The threshold degree is a measure of complexity of Boolean functions (in particular), so that we expect functions with high threshold degree to also have high query complexity. This intuition is validated by the following folklore result: the minimal query complexity of a UPP algorithm that computes $f$ is exactly equal to its threshold degree. We provide a proof of this fact for completeness, as, to the best of our knowledge, it is not explicitly proven in the literature.

We write $f \in \UPP(q)$ when there exists a UPP algorithm with query complexity $q$ that computes $f$, and denote by $q(f)$ the integer such that $f \in \UPP(q(f))$ but $f \notin \UPP(q(f)-1)$.

\begin{lemma}
\label{lem:upp}
    For any $\mathcal{X} \subseteq \fbitset^n$ and $f\colon \mathcal{X} \to \fbitset$, it holds that $\thrdeg(f) = q(f)$.
\end{lemma}
\begin{proof}
    We prove both inequalities, starting with $q(f) \leq \thrdeg(f) \coloneqq d$.

    Let $P(X_1,\ldots,X_n) = \sum_{S \subset [n], \abs{S} \leq d} \alpha_S \prod_{i \in S} X_i$ be a polynomial of degree $d$ that sign-represents $f$, i.e., such that $f(x) = \textsf{sgn}(P(x))$ for all $x \in \mathcal{X}$. Consider the algorithm $A$ that queries the set of coordinates $S$ with probability $\abs{\alpha_S}/\sum_{\abs{T} \leq d} \abs{\alpha_T}$ and outputs $\sgn(\alpha_S) \cdot \prod_{i \in S} x_i$ (note that its query complexity is $d$). Fix $x \in \mathcal{X}$ and suppose, without loss of generality, that $f(x) = 1$. We thus have
    \begin{equation*}
        \E[A^x] = \frac{1}{\sum_{\abs{T} \leq d} \abs{\alpha_T}} \sum_{\abs{S} \leq d} \abs{\alpha_S} \cdot \sgn(\alpha_S) \cdot \prod_{i \in S} x_i = \frac{P(x)}{\sum_{\abs{T} \leq d} \abs{\alpha_T}} > 0,
    \end{equation*}
    and, since $A^x$ only outputs 1 or $-1$, we have $\P[A^x = -1] + \P[A^x = 1] = 1$ and thus $\P[A^x = 1] > 1/2$. It follows that $A$ is a UPP algorithm for $f$ with query complexity $d$ and thus $q(f) \leq \thrdeg(f)$.

    To prove the reverse inequality, consider a UPP algorithm $A$ that computes $f$ with query complexity $q \coloneqq q(f)$, given by a distribution over decision trees of depth at most $q$. To see that the function computed by each decision tree $T$ can be sign-represented by a polynomial of degree at most $q$ (which is a standard fact), we follow the exposition on \emph{leaf indicators} in \cite{GM21}. Denote by $L$ the set of leaves of $T$, and identify each $\ell \in L$ with its indicator function $\ell \colon \fbitset^n \to \bitset$ such that $\ell(x) = 1$ if and only if $\ell$ is the unique leaf reached on input $x$ in $T$.

     Then, if $c_\ell \in \fbitset$ is the output of the decision tree when an execution ends at the leaf $\ell$, the output of $T$ on input $x$ is $\sum_{\ell \in L} c_\ell \cdot \ell(x)$. Thus, showing $\ell(\cdot)$ can be represented by a polynomial of degree at most $q$ implies the same degree bound for the computation of $T$. Fix $\ell \in L$, let $(i_1,\ldots,i_d) \in [n]^q$ be the coordinates queried by the root-to-leaf path that ends at $\ell$, and let the sequence of bits $(b^\ell_1,\ldots,b^\ell_d) \in \fbitset^q$ correspond to the queried values that cause this path to be followed. Then,
    \begin{equation*}
    	\ell(x_1,\ldots,x_n) = \prod_{j = 1}^q \frac{x_{i_j} + b_j}{2 b_j},
    \end{equation*}
    so $\ell(\cdot)$ can be represented by the degree-$q$ polynomial $P(X_1, \ldots, X_n) = 2^{-q} \prod_{j = 1}^q (X_{i_j} + b_j)/b_j$. Thus, the output of $A^x$ \emph{when it selects this tree} is the degree-$q$ polynomial $\sum_{\ell \in L} c_\ell \cdot \ell(x)$, and $\E[A^x]$ is a convex combination of such sums (which also has degree $q$). Since $f(x) = \sgn(\E[A^x])$ for all $x \in \mathcal{X}$, we conclude that $\thrdeg(f) \leq q(f)$ and the claim follows.
\end{proof}

The final ingredient to show the lower bound is the next theorem, a special case of the ``Theorem of the Alternative'' \cite{OS03, ABFR91}.

\begin{theorem}
    \label{thm:alternative}
    Let  $\mathcal{X} \subseteq \fbitset^n$ and let $f\colon \mathcal{X} \to \fbitset$ be any partial Boolean function defined over domain $\mathcal{X}$. If there exists a distribution $\mathcal{D}$ on $\mathcal{X}$ such that $\E_{x \leftarrow \mathcal{D}}[f(x) \cdot m(x)] = 0$ for every monomial $m$ of degree less than $k$, then the threshold degree of $f$ is at least $k$.
\end{theorem}

We are now ready to prove the main result of this section.

\begin{theorem}
    \label{thm:kwise}
    Let $\Pi \subseteq \fbitset^n$ be a $k$-wise independent property such that $\abs{\Pi} < 2^{(1/4 - H(\eps)) n}$ with $k = o(n)$. Then $f \notin \UPP(k-1)$, where $f$ is the partial function such that $f(x) = -1$ when $x \in \Pi$, $f(x) = 1$ when $x$ is $\eps$-far from $\Pi$, and $f$ is
    undefined otherwise.
\end{theorem}
\begin{proof}
    First, apply \Cref{lem:kwisecode} to obtain a $k$-wise independent code $C \subseteq \bitset^n$ that is $\eps$-far from $\Pi$. Let $\mathcal{D}$ be the distribution obtained by drawing a uniform random element of $\Pi$ with probability $1/2$ and  drawing a uniform random element of $C$ with probability $1/2$. Then for every monomial $m$ of degree less than $k$,
    \begin{align*}
    	\E_{x \leftarrow \mathcal{D}}[f(x) m(x)] &= \frac{\E_{x \leftarrow \Pi}[f(x) m(x)] + \E_{x \leftarrow C}[f(x) m(x)]}{2}\\
    	&= \frac{\E_{x \leftarrow \Pi}[m(x)] - \E_{x \leftarrow C}[m(x)]}{2} = 0.
    \end{align*}
    The final equality above holds by virtue of the $k$-wise independence of both $\Pi$ and $C$. Let $\mathcal{X}$
    be the union of inputs in $\Pi$ and inputs that are $\eps$-far from $\Pi$. Define the partial function $f$ over
    domain $\mathcal{X}$ via:
    \begin{equation*}
        f(x) = \left\{\begin{array}{cl}
             -1&\text{, if } x \in \Pi  \\
             1&\text{, if } x \in \mathcal{X} \setminus \Pi.
        \end{array}\right.
    \end{equation*}
    By \Cref{thm:alternative}, the distribution $\mathcal{D}$ constructed above witnesses the fact that $\thrdeg(f) \geq k$. Since the $\UPP$ query complexity of $f$ is $\thrdeg(f)$ by \Cref{lem:upp}, the claim follows.
\end{proof}

We conclude the section with a corollary that follows from the inclusion $\QMA \subseteq \PP$. The proof of this inclusion (in the polynomial-time setting) proceeds in two steps: (1) reducing the error rate of a QMA algorithm to roughly $2^{-p}$, where $p$ is the length of the proof given to the verifier, by repeating the algorithm $O(p)$ times; and (2) running the verifier with the proof fixed to be the maximally mixed state. This exhibits a gap of roughly $2^{-p}$ between the acceptance probabilities of yes- and no-inputs, which suffices to place the problem in $\PP$ \cite{MW05, W09}. The same transformation, carried out \emph{in the query model}, implies that any sufficiently small $\Pi \in \QMAP(\eps,p,q)$ can be ``$\eps$-tested'' by a UPP algorithm with query complexity $O(pq)$; that is, any function $f$ as in the statement of \Cref{thm:kwise} is such that $f \in \UPP(O(pq))$. Therefore,

\begin{corollary}
\label{cor:kwise}
    For any sufficiently constant small $\eps > 0$ and $k$-wise independent property $\Pi \subseteq \bitset^n$ such that $\abs{\Pi} < 2^{n/5}$, we have $\Pi \notin \QMAP(\eps, p,q)$ unless $pq = \Omega(k)$.
\end{corollary}

\section{Quantum speedups for proximity-oblivious MAPs}
\label{sec:po-map}

We now shift gears and move to our \emph{algorithmic} results. We recall in this section the technique of quantum \emph{amplitude amplification}, and prove its consequences for the classes of algorithms we consider in this work. Roughly speaking, given an algorithm that finds, with probability $\gamma$, a preimage of 1 of a Boolean function, amplitude amplification allows us to repeat it $O(1/\sqrt{\gamma})$ times in order to find such a preimage with high probability (as opposed to $O(1/\gamma)$ repetitions classically). Formally, we have:

\begin{theorem}[\cite{BHMT02}]
\label{thm:aa}
	Let $v \colon S \to \bitset$ be a Boolean function (from an arbitrary set $S$) and let $A$ be a quantum algorithm that makes no intermediate measurements (i.e., is a unitary transformation), such that measuring the state $A \ket{\mathbf{0}}$ yields as outcome $s \in v^{-1}(1)$ with probability $\gamma > 0$. Then there exists a quantum algorithm $B$ that uses $O(1/\sqrt{\gamma})$ applications of the unitaries $A$ and $A^{-1}$, such that measuring $B \ket{\mathbf{0}}$ yields as outcome $s \in v^{-1}(1)$ with probability $2/3$.
\end{theorem}

We note that the theorem applies to classical randomised algorithms as a special case. An immediate corollary for \emph{promise problems} in the query model (which is the setting for property testers, MAPs and variations thereof) is the following.\footnote{While we could state amplitude amplification for testers directly, a subtle issue would arise: MAPs are equivalent to a collection of \emph{partial} testers, which are not ``vanilla'' testers but are still promise problems.}

\begin{corollary}[Amplitude amplification for promise problems in the query model]
\label{cor:aa}
	Let $Y, N \subseteq \bitset^n$ with $Y \cap N = \varnothing$ define a promise problem on $n$-bit strings whose yes- and no-inputs are $Y$ and $N$, respectively. Let $A$ be a randomised algorithm with oracle access to a string $x \in \bitset^n$ that makes $q$ queries, always accepts when $x \in Y$ and rejects with probability at least $\gamma$ when $x \in N$. Then, there exists a quantum algorithm $B$ that makes $O(q/\sqrt{\gamma})$ queries to the unitary $U_x \ket{i}\ket{b} = \ket{i}\ket{b \oplus x_i}$, always accepts when $x \in Y$ and rejects with probability $2/3$ when $x \in N$.
\end{corollary}

This follows from the observation that each $x \in Y \cup N$ induces a function $f_x \colon \bitset^r \to \bitset$ where $r$ is the number of random bits used by $A$. If $A^x$ accepts when the outcome of its random coin flips is $s$, we define $f_x(s) = 0$, and if $A^x$ rejects when its random string is $s$, then $f_x(s) = 1$. We then apply \Cref{thm:aa} to the algorithm $A^x$, for each fixed $x \in Y \cup N$ (or, more precisely, to the modified algorithm that computes $f_x$ written as a reversible circuit and thus implements a query to $x$ as $(i,b) \mapsto (i, b \oplus x_i)$), obtaining $B^{U_x}$ (recall that $U_x$ is the unitary mapping $\ket{i}\ket{b} \mapsto \ket{i}\ket{b \oplus x_i}$). Measuring $B^{U_x}\ket{\mathbf{0}}$, using the outcome as the random string for an execution of $A^x$ and outputting accordingly yields the claimed algorithm.

Note that \Cref{cor:aa} directly applies to \emph{one-sided proximity-oblivious} testers, which are testers that always accept $n$-bit strings in the property and reject strings that are $\eps$-far from it with \emph{detection probability} $\rho(\eps,n)$. We now prove the following observation, which shows that the same speedup can be obtained for MAPs; more precisely, properties that admit one-sided proximity-oblivious MAPs allow for more efficient verification by a quantum algorithm using the same proof string.

\begin{theorem}
\label{thm:po-map}
	Let $\Pi$ be a property admitting a one-sided proximity-oblivious MAP protocol, which receives a proof of length $p = p(n)$, makes $q = q(n)$ queries and rejects strings $\eps$-far from $\Pi$ with probability at least $\rho = \rho(\eps, n)$. Then, for any $\eps \in (0,1)$,
	\begin{equation*}
		\Pi \in \QCMAP\left(\eps, p, \frac{q}{\sqrt{\rho}}\right)\;.
	\end{equation*}
\end{theorem}
\begin{proof}
	 Observe that the MAP verifier $V$ can be equivalently described as a collection of probabilistic algorithms $\set{V_\pi : \pi \in \bitset^p}$ indexed by all proof strings $\pi$. By definition, for every $x \in \Pi$ there exists $\pi \in \bitset^p$ such that $V^x_\pi$ always accepts; and, for every $x$ that is $\eps$-far from $\Pi$, every proof string $\pi$ is such that $V^x_\pi$ rejects with probability at least $\rho$. Therefore, $V_\pi$ solves the promise problem whose yes-inputs comprise the subset of $\Pi$ for which $\pi$ is a valid proof, and whose no-inputs are the strings $\eps$-far from $\Pi$.

	 Let $W_\pi$ be the algorithm obtained from $V_\pi$ by \Cref{cor:aa}. Then $W^x_\pi$ accepts (with probability 1) when $x \in \Pi$ and $\pi$ is a valid proof for $x$, and $W^x_\pi$ rejects (with probability $2/3$) when $x$ is $\eps$-far from $\Pi$ and $\pi$ is any proof string; in other words, the algorithm $W$ that executes $W_\pi$ when it receives $\pi$ as a proof string is a QCMAP verifier for $\Pi$. Moreover, since the proof string is reused and $W$ makes $O(q/\sqrt{\rho})$ queries, the proof and query complexities are as stated.
\end{proof}

We conclude with two applications of \Cref{thm:po-map}: to \emph{read-once branching programs} (ROBPs) and \emph{context-free languages} (CFLs), which are shown to admit proximity-oblivious MAPs in \cite{GGR18} (see \Cref{rem:ggr} for details on these results).

\begin{theorem}[\cite{GGR18}, Lemma 3.1]
\label{thm:robp-map}
	For every read-once branching program on $n$ variables of size $s = s(n)$, let $A_B \coloneqq \set{x \in \bitset^n : B(x) = 1}$ be the set of strings accepted by $B$. Then, for every $k \leq n$, the property $\Pi_B$ admits a one-sided proximity-oblivious MAP with communication complexity $O(k \log s)$, query complexity $n/k$ and detection probability $\rho(\eps, n) = \eps$.
\end{theorem}

\begin{theorem}[\cite{GGR18}, Lemma 4.5]
\label{thm:cfl-map}
	For every $k \leq n$, every context-free language $L$ admits a one-sided proximity-oblivious MAP with communication complexity $O(k \log n)$, query complexity $n/k$ and detection probability $\rho(\eps, n) = \eps$.
\end{theorem}

Therefore, applying \Cref{thm:po-map} to \Cref{thm:robp-map,thm:cfl-map}, we obtain:
\begin{corollary}
\label{cor:robp}
	For every read-once branching program $B$ on $n$ variables of size $s = s(n)$, denote by $A_B \coloneqq \set{x \in \bitset^n : B(x) = 1}$ the set of strings accepted by $B$. Then
	\begin{equation*}
		A_B \in \QCMAP\left(\eps, O(k \log s), O\left(\frac{n}{k \sqrt{\eps}}\right)\right) \quad \text{for every } k \leq n \text{ and } \eps \in (0,1).
	\end{equation*}
\end{corollary}

\begin{corollary}
\label{cor:cfl}
	For every context-free language $L$,
	\begin{equation*}
		L \in \QCMAP\left(\eps, O(k \log n), O\left(\frac{n}{k \sqrt{\eps}}\right)\right) \quad \text{for every } k \leq n \text{ and } \eps \in (0,1).
	\end{equation*}
\end{corollary}

Interestingly, these corollaries make explicit a phenomenon in quantum proofs of proximity that does not hold for their classical counterparts: it is possible to test with proximity $\eps = 1/n$, i.e., solve the \emph{exact decision} problem of acceptance by an ROBP and membership in a context-free language, with sublinear proof and query complexity. In particular, taking $k = n^{3/4}$, both complexities are $O(n^{3/4})$. Nonetheless, for the case of branching programs, we will show in \Cref{sec:exact} how to lift the read-once restriction and improve on the parameters by directly exploiting \emph{decomposability}.

\begin{remark}
\label{rem:ggr}
	We note that a context-free language $L$ is defined in terms of an alphabet of \emph{terminals} (which is generally larger than $\bitset$) as well as an alphabet of variables. However, if both alphabets have constant size, we may represent symbols as bit strings with a constant overhead per query; thus \Cref{cor:cfl} holds for languages over large (constant-size) alphabets.

	Moreover, the results of \cite{GGR18} corresponding to \Cref{cor:robp,cor:cfl} are in fact stronger: both apply more generally to IPPs, and thus to MAPs as a special case (see \cite[Section 3.2]{GGR18} for details). In addition, the detection probability of the MAP for ROBPs is $\eps n/n'$ if the branching program has an accepting path of length $n' \leq n$; and the MAP for context-free languages works for \emph{partial derivation languages}, a generalisation of CFLs whose strings may include variable symbols as well as terminals.
\end{remark}

\section{Decomposable properties}
\label{sec:decomposable}

In this section, we show how quantum speedups can be applied to proof of proximity protocols for properties that can be broken up into sub-problems in a distance-preserving manner. Roughly speaking, a property $\Pi$ of $n$-bit strings is \emph{$(k,s)$-decomposable} if, using $s$ bits of information (which we call a \emph{specification}), $\Pi$ can be mapped to $k$ properties $\set{\Lambda^{(i)}}$ and the input string $x$ can be mapped to a set of $k$ strings $\set{x^{(i)}}$ satisfying the following conditions: (1) when $x \in \Pi$, there exists a specification such that $x^{(i)} \in \Lambda^{(i)}$ for all $i \in [k]$; and (2) when $x$ is $\eps$-far from $\Pi$, then, for some specification, $x^{(i)}$ is roughly $\eps$-far from $\Lambda^{(i)}$ for an average $i \in [k]$.

\begin{definition}[Decomposable property]
\label{def:decomposable}
	Let $\Pi = \bigcup \Pi_n$ be a property of bit strings. For $k = k(n)$, $s = s(n)$, $m_1 = m_1(n), \ldots, m_k = m_k(n)$, we say $\Pi$ is \emph{$(k,s)$-decomposable} if there exists a mapping from $S \subseteq \bitset^s$ to (possibly distinct) subproperties $\Lambda^{(1)} \subset \bitset^{m_1}, \ldots, \Lambda^{(k)} \subset \bitset^{m_k}$ such that every $x \in \bitset^n$ uniquely determines $x^{(i)} \in \bitset^{m_i}$ satisfying:\footnote{We remark that the mappings $\Pi \mapsto (\Lambda^{(1)},\Lambda^{(2)},\ldots,\Lambda^{(k)})$ and $x \mapsto (x^{(1)}, x^{(2)},\ldots,x^{(k)})$ are functions of the specification $y \in S$ of the decomposition. Although the notation $x^{(i),y}$ and $\Lambda^{(i),y}$ is formally more accurate, the dependency on $y$ will be clear from context and we omit it for ease of notation.}
	\begin{enumerate}
		\item If $x \in \Pi$, then there exists $y \in S$ such that $x^{(i)} \in \Lambda^{(i)}$ for all $i \in [k]$; and
		\item If $x$ is $\eps$-far from $\Pi$, then, for all $y \in S$ and $i \in [k]$, the string $x^{(i)}$ is $\eps_i$-far from $\Lambda^{(i)}$ and $\E_{i \leftarrow \mathcal{D}}[\eps_i] = \Omega(\eps)$, where $\mathcal{D}$ is the distribution over $[k]$ with probability mass $m_i/(\sum_{j \in [k]} m_j)$ on $i$.
	\end{enumerate}
	If $s = O(k \log n)$ we say $\Pi$ is \emph{succinclty} $k$-decomposable. If the strings $x^{(i)}$ form a partition of $x$, we say $\Pi$ is $(k,s)$-\emph{partitionable}.
\end{definition}

Note that $k$-decompositions specified by $O(k)$ coordinates of the input string are succinct. All of our applications are to succincly decomposable properties, and often the $\set{x^{(i)}}$ form an equipartition of $x$ (so each bit of $x^{(i)}$ depends on a single bit of $x$) and $\mathcal{D}$ is thus the uniform distribution. However, note that if a decomposition is significantly asymmetric, then $\mathcal{D}$ preserves (average) distance while uniform sampling may deteriorate it to $o(\eps)$ (e.g., if $m_i = o(m_1)$ when $i > 1$ and $x^{(1)}$ concentrates all of the corruption).

While the second condition of \Cref{def:decomposable} requires the expectation lower bound to hold for arbitrary $\eps$, the definition is still meaningful when it holds only for restricted values of $\eps$. Indeed, we will make use of it for \emph{exact decision} problems in \Cref{sec:exact}, where the only proximity parameter we consider is $\eps = 1/n$; we call such properties decomposable (or partitionable) with respect to exact decision.

As we will see in the next sections, decomposable properties enable the construction of efficient proof of proximity protocols and generalise the notion of ``parametrised concatenation properties'' introduced by \cite{GR18}.

\subsection{Boosting decompositions via amplitude amplification}
\label{sec:dec-map}

As the next theorem shows, decomposable properties allow for quantum speedups regardless of whether they admit proximity-oblivious MAPs.

\begin{theorem}
	\label{thm:decomposable}
	Let $\Pi$ be a property that is $(k,s)$-decomposable into properties of $m_i$-bit strings, and set $m = \max_{i \in [k]}\set{m_i}$. Suppose each bit of $x^{(i)}$ can be determined by reading $b$ bits of the input string, and each $\Lambda^{(i)}$ admits a one-sided MAP with proximity parameter $\eps$, query complexity $q = q(m,\eps) = m^\alpha/\eps^\beta$ and proof complexity $p = p(m,\eps)$. Then
	\begin{equation*}
		\Pi \in \QCMAP(\eps, s + kp, q')\;,
	\end{equation*}
	with
	\begin{equation*}
		q' =
		\begin{cases}
		  \tilde O\left(b \cdot m^\alpha \cdot \eps^{-\max\left(\frac1{2},\beta\right)}
		  \right) & \text{if }\alpha > 0
		  \text{ and } \beta \geq 0\\
		  \tilde O\left(b \cdot \min\set{m^{1 - \frac1{2\beta}} / \sqrt{\eps},~ m^{1 - \frac1{\beta}} / \eps}
		  \right) & \text{if }\alpha = 0 \text{ and } \beta \geq 1,
		\end{cases}
    \end{equation*}
    where $\tilde O$ hides polylogarithmic factors in $1/\eps$ (but not $m$). Moreover, for \emph{exact decision} (i.e., testing with proximity $\eps = 1/n$), a proof of length $s$ and $O(b m \sqrt{k})$ queries suffice.
\end{theorem}

Before proceeding to the proof, we note that if the MAP protocols for the subproperties are proximity-oblivious, the query complexity can be improved (see \Cref{rem:po-decomposable}). Let us also summarise the proof strategy of \cite{GR18}, which we build upon and generalise.

Consider the special case where a property $\Pi$ is \emph{$k$-partitionable} and the strings $x^{(i)}$ are simply the substrings of $x$ of length $m = n/k$ which, concatenated, form $x$. Suppose, moreover, that the subproperties $\Lambda^{(i)}$ admit \emph{testers} with query complexity $q = m^\alpha/\eps^\beta$ and that $\Pi$ is the union of $\Lambda^{(1)} \times \Lambda^{(2)} \times \cdots \Lambda^{(k)}$ (over all strings in $S$). Note that while, in general, a specification must show how to obtain $\Lambda^{(i)}$ from $\Pi$ \emph{and} how to obtain $x^{(i)}$ from $x$, for an equipartition the latter is implicit.

A natural candidate for a (classical) MAP protocol for $\Pi$ is to guess an index $i \in [k]$ and run the tester for $\Lambda^{(i)}$ on $x^{(i)}$. If $x \in \Pi$, then $x^{(i)} \in \Lambda^{(i)}$ for $i \in [k]$ and the tester always accepts; while if $x$ is $\eps$-far from $\Pi$, then $x^{(i)}$ is $\eps_i$-far from $\Lambda^{(i)}$ for some $\eps_i$ satisfying $\frac1{k} \sum_i \eps_i \geq \eps$, \emph{regardless of the specification} (recall that $x$ is $\eps$-far from $\bigcup_{y \in S} \Lambda^{(1)} \times \Lambda^{(2)} \times \cdots \Lambda^{(k)}$).

We now proceed to the proof of the general case, where the decomposition need not be a partition, and it suffices for the subproperties to admit a MAP (rather than a tester). Moreover, we show that quantum algorithms enable a speedup via amplitude amplification (but this requires the MAPs to be \emph{one-sided}, unlike in the classical case).

\begin{proof}
	Recall that we have a property $\Pi$ that is $(k,s)$-decomposable by a collection of strings $S \subseteq \bitset^s$, where each $y \in S$ determines $k$ properties $\Lambda^{(i)} \subseteq \bitset^{m_i}$ and a decomposition of $x$ into $k$ strings $x^{(i)} \in \bitset^{m_i}$. Moreover, each $\Lambda^{(i)}$ admits a MAP with proof complexity $p$ and query complexity $m^\alpha/\eps^\beta$. The verifier for $\Pi$ executes the steps shown in \Cref{fig:decomposable}.

	We note that a more naive strategy would succeed, albeit with a worse dependence on $\eps$: choosing $i \in [k]$ with probability proportional to $m_i$ yields a string $x^{(i)}$ which is $\eps/2$-far from $\Lambda^{(i)}$ with probability at least $\eps/2$, so that one could execute the MAP verifier for $\Lambda^{(i)}$ with proximity parameter $\eps/2$ (and use amplitude amplification to achieve constant soundness by repeating this $O(1/\sqrt{\eps})$ times). However, the technique of \emph{precision sampling} \cite{L87} overcomes the issue of not knowing the distances $\eps_i$ between $x^{(i)}$ and $\Lambda^{(i)}$ more economically: trying every proximity parameter $2^j$ with $j \in [O(\log 1/\eps)]$, in the spirit of binary search, incurs a merely logarithmic overhead.

  \begin{figure}[h]
    \begin{boxedminipage}{\textwidth}
      \small \medskip \noindent

      \textbf{Input:} explicit access to a proximity parameter $\eps > 0$ and a proof string $\pi \in \bitset^{s + kp}$, and oracle access to a string $x \in \bitset^n$.

      \begin{enumerate}
      \item Interpret the proof as a concatenation of a string $y \in \bitset^s$ with $k$ strings $\pi_1, \ldots, \pi_k \in \bitset^p$. If $y \notin S$, i.e., $y$ does not specify a decomposition, then reject.

      \item For every $j \in [\lceil\log 1/\eps\rceil + 1]$, let $M_j$ be the algorithm obtained from \Cref{cor:aa} by performing $O\left(\sqrt{\frac{\log 1/\eps}{2^j \eps}}\right)$ rounds of amplitude amplification to the following subroutine:
	  \begin{enumerate}[ref =\theenumi{(\alph*)}]
	  	\item\label{itm:aa-subroutine} Sample $i \in [k]$ with probability $m_i/\sum_{j \in [k]} m_j$ and run the MAP verifier for $\Lambda^{(i)}$ on input $x^{(i)}$, with proximity parameter $2^{-j}$, using $\pi_i$ as the proof string. Reject if the MAP for $\Lambda^{(i)}$ rejects.
	  \end{enumerate}
      \item Execute $M_j$ for every $j \in [O(\log 1/\eps)]$. If any of them rejects, then reject; otherwise, accept.
      \end{enumerate}
    \end{boxedminipage}

    \caption{QCMAP verifier for a $(k,s)$-decomposable property $\Pi$}
    \label{fig:decomposable}
  \end{figure}

	Completeness follows immediately: if $x \in \Pi$, then there exists a string $y \in S$ such that $x$ determines $x^{(i)} \in \Lambda^{(i)}$ for all $i \in [k]$. Since the properties $\Lambda^{(i)}$ admit one-sided MAP protocols with proof complexity $p$, the procedure in \Cref{itm:aa-subroutine} always accepts when given the proof string $\pi = (y, \pi_1, \ldots, \pi_k)$, where $\pi_i$ is a valid proof for $x^{(i)}$. Therefore, the verifier always accepts as well.

   Now, suppose $x$ is $\eps$-far from $\Pi$ and the proof string $\pi = (y, \pi_1, \ldots, \pi_n)$ is such that $y \in S$ (since otherwise the verifier rejects immediately). Then, since $\Pi$ is decomposable, $x^{(i)}$ is $\eps_i$-far from $\Lambda^{(i)}$ and $\E_{i \leftarrow \mathcal{D}}[\eps_i] = \Omega(\eps)$, where $\Lambda^{(i)}$ are the subproperties defined by $y$ and $\mathcal{D}$ is the distribution over $[k]$ that samples $i$ with probability proportional to $m_i$.

   To show soundness, we will make use of the following (precision sampling) lemma.
  \begin{lemma}[{\cite[Fact A.1]{G14}}]
  	\label{lem:prec-samp}
    There exists $j \in \left[ \left\lceil\log\frac1{\eps}\right\rceil + 1 \right]$
    such that
    \begin{equation*}
    	\P_{i \leftarrow \mathcal{D}}[\eps_i \geq 2^{-j}] = \Omega\left(\frac{2^j \eps}{\log 1/\eps}\right)\;.
    \end{equation*}
  \end{lemma}

  If the procedure in \Cref{itm:aa-subroutine} samples $i \in [k]$ such that $\eps_i \geq 2^{-j}$, then it rejects with probability $2/3$ (since the MAP has soundness $2/3$). With $j$ as ensured by \Cref{lem:prec-samp}, the probability it samples such an $i \in [k]$ is $\Omega\left(\frac{2^j \eps}{\log 1/\eps}\right)$, so that the probability it rejects is $\frac{2}{3} \cdot \Omega\left(\frac{2^j \eps}{\log 1/\eps}\right) = \Omega\left(\frac{2^j \eps}{\log 1/\eps}\right)$; therefore, the algorithm $M_j$ obtained from \Cref{cor:aa} by $O\left(\sqrt{\frac{\log 1/\eps}{2^j \eps}}\right)$ rounds of amplitude amplification rejects, causing the verifier to also reject, with probability $2/3$.

  We now prove the stated upper bounds on the query complexity. For every $j$, each execution of the MAP verifier for $\Lambda^{(i)}$ makes $q(m,2^{-j})$ queries to $x^{(i)}$, which translate into $b \cdot q(m,2^{-j})$ queries to $x$ (since each query to $x^{(i)}$ can be emulated with $b$ queries to $x$). The total query complexity is therefore

  \begin{equation*}
   \sum_{j \in [\lceil\log{1/\eps}\rceil + 1]} \sqrt{\frac{\log 1/\eps}{2^{j}\eps}} \cdot b \cdot q \left(
    m, 2^{-j} \right) = \tilde{O}\left(\frac{b}{\sqrt{\eps}} \sum_{j \in [\lceil\log{1/\eps}\rceil + 1]} \frac{q \left(
    m, 2^{-j} \right)}{2^{j/2}} \right)\;.
  \end{equation*}
  If $q(m,\eps) = m^\alpha/\eps^\beta$ with $\alpha > 0$ and $\beta \geq 0$, then

  \begin{align*}
    \tilde{O}\left(\frac{b}{\sqrt{\eps}} \sum_{j \in [\lceil\log{1/\eps}\rceil + 1]} \frac{q \left(
    m, 2^{-j} \right)}{2^{j/2}} \right) &= \tilde{O} \left( \frac{b m^{\alpha}}{\sqrt{\eps}} \sum_{j \in [\lceil\log(1/\eps)\rceil + 1]}
                                  2^{j\left(\beta-\frac{1}{2}\right)} \right) \\
                                &= \tilde{O} \left(b m^{\alpha} \eps^{-\max(1/2,\beta)}
                                  \right)\;.
  \end{align*}
  If $\alpha = 0$ and $\beta > 0$, we use the bound $q(m,2^{-j}) \leq m$ for all $\eps$ (from the trivial tester that queries the entire input) to obtain two upper bounds for the query complexity: the first is

  \begin{align*}
    \tilde{O}\left(\frac{b}{\sqrt{\eps}} \sum_{j \in [\lceil\log(1/\eps)\rceil + 1]} \frac{q \left(
    m, 2^{-j} \right)}{2^{j/2}} \right) &= \tilde{O} \left( \frac{b}{\sqrt{\eps}}
                                  \sum_{j \in [\lceil\log(1/\eps)\rceil + 1]} \min \set{
                                  \frac{m}{2^{j/2}}, 2^{j \left(\beta-\frac{1}{2}\right)}} \right) \\
                                &= \tilde{O} \left( \frac{b}{\sqrt{\eps}} \sum_{j \in [\lceil\log(1/\eps)\rceil + 1]} m^{1-\frac1{2\beta}} \right)\\
                                &= \tilde{O} \left( \frac{b m^{1-\frac1{2\beta}}}{\sqrt{\eps}} \right)\;.
  \end{align*}

  The second upper bound matches that obtained by the classical MAP, and is tighter when $m^{\frac1{2\beta}} \geq 1/\sqrt{\eps}$. Since $2^{j/2} \leq 4/\sqrt{\eps}$ for all $j$ in the sum above, we have $\sqrt{\eps} = O(2^{-j/2})$; therefore,
    \begin{align*}
    \tilde{O}\left(\frac{b}{\sqrt{\eps}} \sum_{j \in [\lceil\log(1/\eps)\rceil + 1]} \frac{q \left(
    m, 2^{-j} \right)}{2^{j/2}} \right) &= \tilde{O}\left(\frac{b}{\eps} \sum_{j \in [\lceil\log(1/\eps)\rceil + 1]} \frac{\sqrt{\eps}}{2^{j/2}} \cdot q \left(
    m, 2^{-j} \right) \right) \\
    &= \tilde{O} \left( \frac{b}{\eps}
                                  \sum_{j \in [\lceil\log(1/\eps)\rceil + 1]} \min \set{
                                  \frac{m}{2^j}, 2^{j (\beta-1)}} \right) \\
                                &= \tilde{O} \left( \frac{b}{\eps} \sum_{j \in [\lceil\log(1/\eps)\rceil + 1]} m^{1-\frac1{\beta}} \right)\\
                                &= \tilde{O} \left( \frac{b m^{1-\frac1{\beta}}}{\eps} \right)\;.
  \end{align*}

  Finally, observe that, for testing with $\eps = 1/n$ (i.e., deciding exactly), one may take the MAPs for $\Lambda^{(i)}$ to be the trivial testers (with query complexity $m$ and no proof). Moreover, it is unnecessary to iterate over $j$ and apply \Cref{lem:prec-samp}; sampling $i \in [k]$ uniformly and running the trivial tester requires $b m$ queries to $x$ and leads to a rejection with probability at least $1/k$, since $\frac1{k} \sum_{i \in [k]} \eps_i > 0$ implies $x^{(i)} \notin \Lambda^{(i)}$ for at least one $i \in [k]$. Therefore, applying $O(\sqrt{k})$ rounds of amplitude amplification to this procedure ensures rejection of an $x$ that is $\eps$-far from $\Pi$ with constant probability and yields query complexity $O(bm\sqrt{k})$.
\end{proof}

\begin{remark}[Additional speedup for proximity-oblivious MAPs]
\label{rem:po-decomposable}
	If the MAP verifiers for the subproperties $\Lambda^{(i)}$ are proximity-oblivious, it is possible to improve on the query complexity of \Cref{thm:decomposable} significantly. More precisely, suppose each $\Lambda^{(i)}$ admits a proximity-oblivious MAP with query complexity $O(1)$ and detection probability $\rho(\eps,m) = \eps^\beta/m^\alpha$. Then, if an input is $\eps$-far from $\Pi$, for some $j \in [O(\log 1/\eps)]$ (as ensured by \Cref{lem:prec-samp}), the procedure of \Cref{itm:aa-subroutine} samples $i \in [k]$ such that $x^{(i)}$ is $\eps_i$-far from $\Lambda^{(i)}$ with $\eps_i = \tilde{\Omega}(2^{j} \eps)$.

	The procedure thus rejects with probability $\tilde{\Omega}(2^j \eps)\cdot \rho(2^{-j},m) = \tilde{\Omega}(2^{j(1-\beta)} \eps/m^\alpha)$ in this case. By applying $\tilde{O}(m^{\alpha/2}/\sqrt{2^{-j(\beta - 1)} \eps})$ rounds of amplitude amplification for each $j$ (instead of $\tilde{O}(1/\sqrt{2^j \eps})$ performed for MAPs that are not proximity-oblivious), and, since $2^{-j}  =\Omega(\eps)$, the total query complexity becomes $\tilde{O}(b \sqrt{m^\alpha/\eps^\beta})$.
\end{remark}

We finish this section with a corollary of \Cref{thm:decomposable} in the \emph{graph orientation model}. A directed graph is called \emph{Eulerian} if the in-degree of each of its vertices is equal to its out-degree. An orientation is a mapping from the edges to $\bitset$, representing whether each edge is oriented from $i$ to $j$ or from $j$ to $i$, and the distance between two orientations is the fraction of edges whose orientation must be changed to transform one into the other. Let $\Pi_E$ be the property consisting of all Eulerian orientations of the complete bipartite graph $K_{2,n-2}$, i.e., the graph with vertex set $[n]$ and edge set $\set{\set{i,j} : i \leq 2, j \geq 3}$.

A MAP protocol for $\Pi_E$ parametrised by an integer $k$ is constructed in \cite{GR18}; its proof and query complexities are $O(k \cdot \log{n})$ and $\tilde{O}(n/(\eps k))$, respectively. That protocol is obtained by applying their (classical) version of \Cref{thm:decomposable} using the trivial tester for each of the subproperties. Using the same proof (and thus the same decomposition), and recalling that the trivial tester makes $q(m,\eps) = m$ queries for a sub-property of length $m$, we obtain

\begin{corollary}
\label{cor:eulerian}
  The property $\Pi_E$ has a one-sided QCMAP, with
  respect to proximity parameter $\eps$, that uses a proof of length
  $O(k \cdot \log{n})$ and has query complexity $\tilde{O} \left(
    \frac{n}{k\sqrt{\eps}} \right)$.
\end{corollary}

Note that the query complexity of the classical MAP becomes linear with $\eps = \Omega(1/k)$, whereas the QCMAP is able to decide \emph{exactly} (i.e. test with $\eps = 1/n$) with query complexity $O(n^{3/2} / k)$ (which is still sublinear whenever $k = \omega(\sqrt{n})$). In particular, with $k = n^{3/4}$, both query and communication complexities are $\tilde{O}(n^{3/4})$; see the \Cref{sec:exact} for further discussion and applications of \Cref{thm:decomposable} to exact decision problems.

\subsection{\texorpdfstring{$k$}{k}-monotonicity}
\label{sec:k-monotonicity}

In this section, we show that a generalisation of monotonicity of Boolean functions over the line $[n]$ is efficiently testable by QCMAP protocols. A function $f: [n] \to \bitset$ is \emph{$k$-monotone} if any sequence of integers $1 \leq x_1 < x_2 \ldots < x_\ell \leq n$ such that $f(x_1) = 1$ and $f(x_i) \neq f(x_{i+1})$ for all $i < \ell$ has length $\ell \leq k$. This problem was studied in \cite{CGGKW19}, where one-sided $\eps$-testers for $k$-monotonicity on the line are shown to require $\Omega(k/\eps)$ queries, while two-sided testers can achieve query complexity $\tilde{O}(1/\eps^7)$ (which, although a far worse dependence on $\eps$ than the lower bound, is \emph{independent} of $k$).

In order to apply \Cref{thm:decomposable}, we must only show that $k$-monotone functions are decomposable. Define $\Pi_{k,[n]}$ as the set of $k$-monotone Boolean functions on the line $[n]$.\footnote{We consider the standard representation of a Boolean function as the bit string obtained by concatenating all function evaluations, i.e., $f \colon [n] \to \bitset$ is represented by $x \in \bitset^n$ with $x_i = f(i)$.} Then,

\begin{theorem}
	For any $k \in [n]$, the property $\Pi_{k,[n]}$ is succinctly $k$-decomposable.
\end{theorem}
\begin{proof}
	Since a $k$-monotone function $f$ has at most $k - 1$ critical points, where it changes from nondecreasing to nonincreasing or vice-versa, specifying these points yields a decomposition of $f$ into (1-)monotone subfunctions.

	More precisely, a string of length $s \leq (k-1) \log n$ determines a decomposition of the input $f$ by specifying $\ell \leq k - 1$ integers $1 < n_1 < n_2 < \cdots < n_\ell < n$. Define $n_0 = 1$, $n_{\ell + 1} = n$, $m_i \coloneqq n_i - n_{i-1} + 1$ and the function $f_i: [m_i] \to \bitset$ by $f_i(x) = f(x + n_{i-1} - 1)$ for all $i \in [\ell + 1]$. Then, $f \in \Pi_{k,[n]}$ if and only if:
	\begin{enumerate}
		\item for $i \in [\ell + 1]$ odd, $f_i \in \Lambda^{(i)} \coloneqq \set{g \in \Pi_{1, m_i} : g(1) = f_i(1)}$; and
		\item for $i \in [\ell + 1]$ even, $f_i \in \Lambda^{(i)} \coloneqq \set{1 - g : g \in \Pi_{1, m_i} \text{ and } g(1) = 1 - f_i(1)}$.
	\end{enumerate}

	It is clear that, when $f \in \Pi_{k,[n]}$, there exists a set of $\ell \leq k - 1$ distinct integers in $[2,n-1]$ that satisfies both conditions. When $f$ is $\eps$-far from $\Pi_{k,[n]}$, the sum of absolute distances $\eps_i m_i$ from each $f_i$ to $\Lambda^{(i)}$ is $\sum_i \eps_i m_i \geq \eps n$. Therefore,
	\begin{equation*}
		\E_{i \leftarrow \mathcal{D}}[\eps_i] \geq \eps \cdot \frac{n}{\sum_{i \in [\ell + 1]} m_i} = \eps \cdot \frac{n}{n + \ell} = \Omega(\eps)\;,
	\end{equation*}
	where $\mathcal{D}$ is the distribution over $[\ell + 1]$ that has probability mass $m_i/(\sum_{j \in [\ell + 1]} m_j)$ at point $i$.

	Finally, while this ensures $(\ell+1)$-decomposability for some $\ell \leq k - 1$, one may deterministically transform it into a $k$-decomposition (and, in fact, a $K$-decomposition for any $K \geq \ell$) by, for example, iteratively finding the largest interval and dividing it at its midpoint $k - \ell - 1$ times (with the requirement that nonincreasing functions in the large interval are monotone nonincreasing in both subintervals, and likewise for the nondecreasing case).
\end{proof}

Since monotonicity on the line $[m]$ is $\eps$-testable with $q(m,\eps) = O(1/\eps)$ queries \cite[Proposition 1.5]{G17}, applying (the second case of) \Cref{thm:decomposable} yields \Cref{cor:kmonotone}: a QCMAP protocol for $k$-monotonicity with proof complexity $O(k \log n)$ and query complexity $\tilde{O}(1/\eps)$ exists. (The theorem also gives an upper bound of $\tilde O(\sqrt{n/\eps})$, which is no better: $\tilde O(1/\eps)$ is smaller up to $\eps \approx 1/n$, where the two bounds match.)

It is worth noting that the standard monotonicity tester on the line is \emph{not} proximity-oblivious, unlike, e.g., on the Boolean hypercube, where both the ``edge tester'' \cite{GGLRS00} and the state-of-the-art \cite{KMS18} are proximity-oblivious; thus, one could not directly apply amplitude amplification, and must exploit decomposability via \Cref{thm:decomposable}.

Note that the (one-sided) QCMAP for $k$-monotonicity is more efficient than any one-sided tester, e.g., when $k = \Theta(\log n)$ and $\eps = \Theta(1/\log^2 n)$: then the QCMAP's proof and query complexities are $\tilde O(\log^2 n)$, whereas one-sided testers must make $\Omega(\log^3 n)$ queries. Moreover, our QCMAP outperforms the best known \emph{two-sided} tester of \cite{CGGKW19}, which makes $\tilde{O}(1/\eps^7)$ queries, even with mild dependencies of the proximity parameter $\eps$ on $n$. Indeed, when $\eps = o(1/\log^{1/7} n)$ and $k = O(1)$, the tester's query complexity is superlogarithmic while the proof and query complexities of the QCMAP are logarithmic.

With the transformation from two-sided testers to one-sided MAPs of \cite{GR18}, it is also possible to compare our one-sided QCMAP against one-sided MAPs. From the aforementioned two-sided tester, one obtains a MAP with proof complexity $\polylog n$ and query complexity $O(\polylog(n/\eps)/\eps^7)$. Thus, the QCMAP remains advantageous except when $k$ and $\eps$ are large (e.g., $k = \omega(\polylog n)$ and $\eps = \Omega(1)$).

\subsection{Exact problems}
\label{sec:exact}

To conclude the discussion of decomposability and its consequences, we shift focus to a special case: that of testing $n$-bit strings with proximity parameter $\eps = 1/n$. Since, for any $\Pi \subseteq \bitset^n$ and $x \in \bitset^n \setminus \Pi$, the string $x$ is at least $1/n$-far from $\Pi$, this is the task of \emph{exactly} deciding membership in $\Pi$.

Observe that for classical MAPs, nontrivial properties require $\Omega(n)$ queries in this case: even if a verifier receives as proof a claim $x'$ that is allegedly equal to its input string, it requires $O(1/\eps) = \Omega(n)$ queries to check the validity of the claim. Remarkably, \emph{quantum} algorithms are able to solve exact decision problems with sublinear queries (as illustrated by Grover's algorithm, which makes $O(\sqrt{n})$ queries). Thus, as we show next, insights arising from decomposability are applicable to this setting. We begin by showing a nontrivial QCMA protocol for the parity of a bit string in \Cref{sec:parity}, and then extend it to \emph{branching programs} in \Cref{sec:bp}.

\subsubsection{Parity}
\label{sec:parity}

Consider the problem of deciding if an $n$-bit string has even parity. This is clearly maximally hard, requiring $\Omega(n)$ queries even for \emph{interactive proofs with arbitrary communication}, which we show next for completeness.

\begin{lemma}
    Any IP verifier that accepts strings of even parity and rejects strings of odd parity with probability $2/3$ must make at least $n/3$ queries to its input.
\end{lemma}
\begin{proof}
    Let $V$ and $P$ be a verifier and an honest prover for an IP for parity, and assume, towards contradiction, that the query complexity of $V$ is less than $n/3$.

   Fix an arbitrary input $x \in \bitset^n$ of even parity. Define $S_x$ as the random variable comprising all the coordinates queried by $V$ in an execution $\langle V^x, P(x)\rangle$ of the protocol, and let $I \in [n]$ be a uniform random variable independent from $S_x$. Then, since $\abs{S_x} < n/3$,
   \begin{equation*}
    \frac{1}{n} \sum_{i = 1}^n \P[i \in S_x] = \sum_{i = 1}^n \P[I = i] \cdot \P[I \in S_x ~|~ I = i] = \P[I \in S_x] < \frac{1}{3}\;,
   \end{equation*}
   so there exists $i \in [n]$ such that
   \begin{equation*}
    \P[V^x \text{ queries } i \text{ in the execution } \langle V^x, P(x)\rangle] = \P[i \in S_x] < \frac{1}{3}\;.
   \end{equation*}
  Now, consider the execution of a protocol on input $y \in \bitset^n$ obtained by flipping the $i^\text{th}$ bit of $x$ (i.e., such that $y_j = x_j$ if $j \neq i$ and $y_i = 1 - x_i$ otherwise). Let $\tilde{P}$ be a (malicious) prover that executes on $y$ exactly as $P$ does on $x$; that is, set $\tilde{P}(y) = P(x)$. We thus have
  \begin{align*}
    \P[\langle V^y, \tilde{P}(y) \rangle \text{ rejects}] &= \P[i \in S_y] \cdot \P[\langle V^y, \tilde{P}(y) \rangle \text{ rejects} ~|~ i \in S_y]\\
    &+ \P[i \notin S_y] \cdot \P[\langle V^y, \tilde{P}(y) \rangle \text{ rejects} ~|~ i \notin S_y]\;,
  \end{align*}
  and, moreover, the following equalities between events hold:
  \begin{align*}
    [i \notin S_y] &= [i \notin S_x]\;, \text{ and thus}\\
    [\langle V^y, \tilde{P}(y) \rangle \text{ rejects} ~|~ i \notin S_y] &= [\langle V^x, P(x) \rangle \text{ rejects} ~|~ i \notin S_x]\;.
  \end{align*}
  Therefore,
  \begin{align*}
    \P[\langle V^y, \tilde{P}(y) \rangle \text{ rejects}] &< \frac{1}{3} + \P[i \notin S_y] \cdot \P[\langle V^y, \tilde{P}(y) \rangle \text{ rejects} ~|~ i \notin S_y]\\
    &= \frac{1}{3} + \P[\langle V^x, P(x) \rangle \text{ rejects and } i \notin S_x]\\
    &\leq \frac{1}{3} + \P[\langle V^x, P(x) \rangle \text{ rejects}] \leq \frac{2}{3} \;,
  \end{align*}
  contradicting the correctness of the protocol.
\end{proof}

Rather surprisingly, however, there exists a quantum \emph{non-interactive} protocol that exploits amplitude amplification and achieves sublinear query and communication complexities. This is a direct consequence of the following.
\begin{proposition}
	For any $k \leq n$, the property $\Pi \coloneqq \set{x \in \bitset^n : \bigoplus_{j \in [n]} x_j = 0}$ is succinctly $k$-partitionable (with respect to exact decision).
\end{proposition}

\begin{proof}
	The set of strings that specify decompositions is $S = \set{y \in \bitset^k : \bigoplus_{i \in [k]} y_i = 0}$, i.e., the set of $k$-bit strings of even parity. A string $y \in S$ specifies a decomposition where, for each $i \in [k]$, the $i^\text{th}$ subproperty is $\Lambda^{(i)} = \set{x^{(i)} \in \bitset^{n/k} : \bigoplus_{j \in [n/k]} x^{(i)}_j = y_i}$ and $x \in \bitset^n$ induces $x^{(i)}$ as the substring of $x$ consisting of the $i^\text{th}$ block of $n/k$ bits, i.e., $x^{(i)} = \left(x_{\frac{(i-1)n}{k} + 1}, x_{\frac{(i-1)n}{k} + 2}, \ldots, x_{\frac{in}{k}}\right)$.

	Note that the condition $x^{(i)} \in \Lambda^{(i)}$ for all $i \in [k]$ uniquely defines $y \in \bitset^k$ by $y_i = \bigoplus_{\frac{(i-1)n}{k} < j \leq \frac{in}{k}} x_j$. Therefore, if $x$ has even parity, there exists $y \in S$ satisfying the condition, while if $x \notin \Pi$ the only string that satisfies it is not in $S$. Thus $x$ is $1/n$-far from $\Pi$ and, \emph{for all $y \in S$}, there exists $i \in [k]$ such that $x^{(i)} \notin \Lambda^{(i)}$, so that the distances $\eps_j$ from $x^{(j)}$ to $\Lambda^{(j)}$ satisfy $\frac1{k} \sum_{j \in [k]} \eps_j \geq \frac{\eps_i}{k} \geq \frac1{k \cdot \frac{n}{k}} = 1/n$.
\end{proof}

Applying \Cref{thm:decomposable} with $k = n^{2/3}$, we have:

\begin{corollary}
	There exists a QCMA protocol for parity with $O(n^{2/3})$ query and communication complexities.
\end{corollary}

\subsubsection{Branching programs}
\label{sec:bp}

Recall that \Cref{cor:robp} shows membership in the set of strings accepted by \emph{read-once} branching programs can be decided with $O(n^{3/4})$ query and proof complexities; thus, since parity is computed by an ROBP (of width 2), we obtain a QCMA protocol with the same parameters. Observe, however, that this is significantly worse than the $O(n^{2/3})$ upper bound of the previous section, which directly exploits decomposability. This suggests a similar improvement may be possible for branching programs, which we now show to be indeed the case for \emph{layered} branching programs parametrised by their \emph{length} $\ell$ and \emph{width} $w$ (see \Cref{sec:prelims} for the definitions).

\begin{proposition}
\label{prop:bp}
	Let $B$ be a layered branching program on $n$-bit strings of length $\ell = \ell(n)$ and width $w = w(n)$. For any $k \leq \ell$, the set $A_B \subseteq \bitset^n$ of strings accepted by $B$ is $(k, O(k \log w))$-partitionable (with respect to exact decision).
\end{proposition}
\begin{proof}
	Let $V$ be the vertex set of the graph that defines the branching program $B$, whose layers are $V_0 = \set{v_0}$, $V_1$, \ldots, $V_\ell$ and whose set of accepting nodes is $F \subseteq V_\ell$. Decompositions are specified by the set $S = \set{(v_1, \ldots, v_k) : v_i \in V_{i \ell/k} \text{ and } v_k \in F}$ (whose elements are the $k \log w$-bit representations of the sequence of vertices). (Note that $\abs{V_i} \leq w$ implies $\log w$ bits suffice to identify each vertex in a given layer.)  A specification fixes $k$ nodes that partition an alleged accepting path into paths of length $\ell/k$. Thus, $\Lambda^{(i)}$ is the set of strings accepted by the branching program that is the subgraph of $B$ with layers $\set{v_{i-1}}$ and $V_j$ for $\frac{(i-1)\ell}{k} < j \leq \frac{i\ell}{k}$, source node $v_{i-1}$ and accepting node $v_i$. The string $x^{(i)}$ is the $i^\text{th}$ block of $\ell/k$ bits of the input $x$.

	If $x \in A_B$, there clearly exists an accepting path $(v_0, u_1, \ldots, u_\ell)$, namely the path determined by the execution of $B$ on $x$.  Therefore, $y = (u_{\ell/k}, u_{2 \ell/k}, \ldots, u_{(k-1)\ell/k}, u_\ell) \in S$ is a specification satisfying $x^{(i)} \in \Lambda^{(i)}$ for all $i \in [k]$. If $x \notin A_B$, on the other hand, every specification is such that at least one sub-path does not match the corresponding sub-path determined by $x$ (since otherwise $x$ would be accepted by $B$), implying $x^{(i)} \notin \Lambda^{(i)}$ for some $i$.
\end{proof}

Applying \Cref{thm:decomposable} with $k = \ell^{2/3}$, we obtain \Cref{thm:bp}: width-$w$, length-$\ell$ branching programs admit QCMAP protocols with proof complexity $O(\ell^{2/3} \log w)$ and query complexity $O(\ell/\sqrt{k}) = O(\ell^{2/3})$. Note that the $\Omega(n)$ classical complexity lower bound for parity implies the same bound for width-2 branching programs of length $n$, and that the complexities of the QCMA protocol are sublinear up to length $o(n^{3/2}/\log w)$; in particular, this holds for width-$2^{n^\alpha}$ branching programs of length $o(n^{\frac{3}{2}(1 - \alpha)})$ for every $\alpha \geq 0$. Thus, besides lifting the read-once restriction, it improves on \Cref{cor:robp} for a wide range of parameters.

\section{Bipartiteness in bounded-degree graphs}
\label{sec:bipartite}

In the bounded-degree graph model, an algorithm is given query access to the \emph{adjacency list} of a graph $G$ whose vertices have their degree bounded by $d = O(1)$. More precisely, given a vertex $v$ and an index $i \in [d]$, the oracle corresponds to the mapping $(v,i) \mapsto w$, where $w$ is the $i^\text{th}$ neighbour of $v$ (if it exists) or $w = \bot$ (if $v$ has fewer than $i$ neighbours). The distance between two graphs is the fraction of pairs $(v,i)$ whose outputs differ between the adjacency list mappings.

Our goal in this section will be to construct a QMAP protocol for testing whether a bounded-degree \emph{rapidly-mixing} graph $G$ (i.e., where the last vertex in a random walk of sufficiently large length $\ell = O(\log n)$ starting from any vertex is distributed roughly uniformly), given as an adjacency list oracle, is bipartite or $\eps$-far from every (rapidly-mixing) bipartite graph. Note that this differs from the standard testing setting on graphs by the additional restriction that \emph{both} yes- and no-inputs be rapidly mixing: whether $G$ is bipartite or $\eps$-far from bipartite, for any pair $u,v$ of vertices, the length-$\ell$ random walk starting from $u$ ends at $v$ with probability between $1/(2n)$ and $2/n$.

Our QMAP protocol for bipartiteness builds on the MAP protocol of \cite[Theorem 7.1]{GR18}, modifying it to take full advantage of quantum speedups. In the strategy laid out in that work, the classical verifier receives as proof a set $S$ of $k$ vertices that are allegedly on the same side of a bipartition. It repeatedly samples a uniformly random vertex, takes many lazy random walks of length $\ell$ and,\footnote{A lazy random walk moves from a vertex $v$ to a uniform random neighbour with probability $\frac{d_v}{2d}$ and otherwise stays at $v$, where $d_v$ is the degree of $v$ and $d = O(1)$ is the graph's degree bound.} for each of them, records if it stops at a vertex in $S$ as well as the parity of the number of non-lazy steps (where the walk moves to another vertex). If two walks starting from the same vertex end in $S$ with different parities, the verifier has found a witness to the fact that the graph is not bipartite and rejects.

For any proximity parameter $\eps$ and $k \leq n/2$, the MAP protocol of \cite{GR18} requires a proof of length $k \log n$ and makes $O(\eps^{-2} \cdot n/k)$ queries. By making the verifier quantum (but using the same classical proof), we obtain improvements in the dependence on both $n/k$ and $\eps$:

\begin{theorem}
\label{thm:bipartite}
    For every $k \leq n/2$, there exists a one-sided QCMAP for deciding whether an $n$-vertex bounded-degree graph $G$ is bipartite or $\eps$-far from being bipartite, under the promise that $G$ is rapidly-mixing, with proof complexity $k \cdot \log n$ and query complexity $\tilde{O}((n/k)^{2/3} \cdot \eps^{-5/6})$.
\end{theorem}

Our strategy to obtain a quantum speedup uses the quantum \emph{collision-finding} algorithm \cite{A07, ACL11} to a suitable modification of the verifier strategy outlined above.\footnote{We remark that collision-finding here refers to a generalisation of the \emph{element distinctness} algorithm to symmetric relations beyond equality. In particular, it is \emph{not} a quantum algorithm for the ``$r$-to-1 collision problem'', where (many) collisions are promised to exist.} We note that the function $f$ to which we will apply the collision-finding algorithm is \emph{not} the adjacency list oracle, but one whose unitary representation can be obtained by adjacency list queries and classical computation.

\begin{theorem}[Quantum collision-finding {\cite[Theorem 9]{ACL11}}]
	\label{thm:collision}
	Let $f \colon X \to Y$ be a function given via a unitary oracle and $R \subseteq Y \times Y$ be a symmetric binary relation. There exists a quantum algorithm that makes $O(\abs{X}^{2/3} \polylog \abs{Y})$ queries to $f$, always accepts if $(x,x') \notin R$ for every distinct $x, x' \in X$, and rejects with probability $2/3$ if there exist distinct $x,x' \in X$ such that $(x,x') \in R$.
\end{theorem}

Returning to the aforementioned classical MAP, observe that the information of each set of $t$ random walks of length $\ell$ starting from a fixed vertex can be represented by a function $f \colon T \to \bitset^2$, where $T \subset \bitset^{\ell'}$ is composed of $t$ strings of length $\ell' = O(\ell)$ (note that each step of the random walk can be performed with $O(1)$ bits, for a total of $\ell' = O(\ell)$ random coins per walk) and $f(r)$ encodes whether the walk reaches a vertex in $S$ as well as its parity when its inner randomness corresponds to the string $r$. Now, applying the collision-finding algorithm to this function for each sampled vertex already yields a speedup \cite{ACL11}; but it can be improved by combining this strategy with amplitude amplification (which improves quadratically the number of samples taken as the starting vertices of the random walks), as we show next.

\begin{proof}[Proof of \Cref{thm:bipartite}]
    We first consider the classical MAP outlined earlier when it samples a \emph{single} starting vertex for a random walk, and apply collision-finding to the function $f$ whose domain is $T$, a random (multi-)set of sequences of coin flips. The proof specifies a subset $S \subseteq L$ of size $k$, for some $L$ (allegedly) defining a bipartition $V = L \cup R$ of $G$ with $\abs{L} \geq \abs{R}$. The resulting algorithm (which is \emph{not} yet the final QCMAP verifier), is shown in \Cref{fig:bipartiteness}.

  \begin{figure}[h]
    \begin{boxedminipage}{\textwidth}
      \small \medskip \noindent

      \textbf{Input:} oracle access to the adjacency list of an $n$-vertex graph $G$, as well as explicit access to a proximity parameter $\eps > 0$ and a proof string $\pi \in \bitset^{k \log n}$ (representing a set $S \subset V(G)$ of size $k$).

      \begin{enumerate}

      \item Sample a uniformly random  vertex $v \in V(G)$ and select a multi-set $T \subset \bitset^{\ell'}$, where $\ell' = O(\ell)$ and $\ell = O(\log n)$, by sampling $O \left( \frac{n}{k} \cdot \frac{\log{n}}{ \eps} \right)$ bit strings uniformly and independently.

      \item Let $f: T \to \bitset^2$ be the function computed by the following subroutine:
	  \begin{enumerate}[ref =\theenumi{(\alph*)}]
	  	\item\label{itm:cf-subroutine} Let $r \in T$ be the input. Take the (lazy) random walk of length $\ell$ starting at $v$ using $r$ as the random bits. Let $a \in \bitset$ be the indicator of whether the walk stops at a vertex in $S$ and $b \in \bitset$ be the parity of the walk (i.e., the number of non-lazy steps). Return $(a, b)$.
	  \end{enumerate}
      \item Execute the algorithm of \Cref{thm:collision} with respect to $f$ to find $t$ and $t'$ such that $f(t) = (a,b)$ and $f(t') = (a',b')$ with $a = a' = 1$ and $b \neq b'$. Reject if such a pair is found, and accept otherwise.
      \end{enumerate}
    \end{boxedminipage}

    \caption{Quantum algorithm for bipartiteness with low detection probability}
    \label{fig:bipartiteness}
  \end{figure}

	First, note that the function $f$ contains a collision with respect to the relation $R \subset \bitset^4$ that comprises all pairs of pairs of the type $((1,b),(1,1-b))$ \emph{if and only if} there are two random walks of length $\ell$ that start from $v$ end at a vertex in $S$ with different parities.

    If $G$ is bipartite with vertex set $L \cup R$ and $S \subseteq L$, every path with both endpoints in $S$ has even length. But the existence of two paths from the same vertex into $S$ with different parities implies the existence of a path of odd length with both endpoints in $S$; therefore, since the function $f$ does not contain a collision for any starting vertex $v$, the verifier always accepts when $G$ is bipartite and the proof consists of $k$ elements on the same side of a bipartition.

    If $G$ is $\eps$-far from bipartite, then \cite{GR18} (building on \cite{GR99}) show the following: defining $a$ as the indicator of whether a random walk of length $\ell = O(\log n)$ starting from vertex $v$ stops at a vertex in $S$, and $b$ as the parity of this random walk, then an $\Omega(\eps/\log n)$ fraction of vertices in $V(G)$ are such that \emph{both} $\P[a = 1, b = 0]$ and $\P[a = 1, b = 1]$ are $\Omega(\frac{k \eps}{n \log n})$. Therefore, with probability $\Omega(\eps/\log n)$, the sampled vertex $v \in V(G)$ satisfies this condition. If such a vertex was sampled, with probability $\Omega(1)$, in a sufficiently large number $\abs{T} = O(\frac{n \log n}{k \eps})$ of random walks starting from it there exists a pair of length-$\ell$ walks that end in $S$ with different parities. Finally, since the algorithm of \Cref{thm:collision} rejects with constant probability if the function $f$ contains a collision, the algorithm in \Cref{fig:bipartiteness} rejects with probability $\Omega(\eps/\log n)$.

    The quantum collision-finding algorithm computes $O(\abs{T}^{2/3})$ times the function $f$, each of which simulates a random walk of length $\ell = O(\log n)$. Each step of the random walk requires $O(1)$ queries to the graph $G$, so that $O(\log n)$ queries are made per walk. The query complexity of the algorithm in \Cref{fig:bipartiteness} is therefore
    \begin{equation*}
    	O\left(\abs{T}^{2/3} \cdot \ell\right) = \tilde{O}\left( \left(\frac{n}{k \eps}\right)^{2/3} \right).
    \end{equation*}

    Finally, applying amplitude amplification (\Cref{thm:aa}) to this algorithm (rather, more precisely, to a reversible quantum algorithm that is obtained from \Cref{fig:bipartiteness} by deferring its measurements; see, e.g., \cite{NC16}), we obtain a QCMAP verifier for bipartiteness of bounded-degree graphs with query complexity
     \begin{equation*}
    	\tilde{O}\left( \left(\frac{n}{k \eps}\right)^{2/3} \right) \cdot O\left(\sqrt{\frac{\log n}{\eps}}\right) = \tilde{O}\left( \left(\frac{n}{k}\right)^{2/3} \cdot \frac1{\eps^{5/6}} \right).\qedhere
    \end{equation*}
\end{proof}

We remark that classically testing bipartiteness in the bounded-degree model requires $\Omega(\sqrt{n})$ queries \emph{even under the rapidly-mixing promise}, as shown by Goldreich and Ron \cite{GR02}; therefore, for constant $\eps$, a QCMAP with proof complexity $\tilde{O}(n^{1/4})$ is enough to overcome the testing lower bound, while the MAP of \cite{GR18} requires a proof of length $\tilde{O}(\sqrt{n})$. Of course, a fairer comparison would be to \emph{quantum} testers, for which $\tilde{O}(n^{1/3})$ queries suffice \cite{A07} but no nontrivial lower bound is known. The quantum tester is sound without the rapidly-mixing promise, however, so even though a QCMAP with proof length $\tilde{O}(\sqrt{n})$ makes fewer queries than the best known quantum tester, it is not clear how the two algorithms compare.


\appendix

\section{A QCMAP lower bound for testing unitaries}
\label{sec:qmap-vs-qcmap}

The known $\QMA$ versus $\mathcal{QCMA}$ (oracle) separation of Aaronson and Kuperberg \cite{AK07} is naturally cast as a testing problem, and thus yields a corresponding separation between $\QMAP$ and $\QCMAP$. More precisely, we consider the following property of unitaries:
\begin{equation*}
    \Pi_{AK} = \set{U \in \mathbb{C}^{n \times n} : \begin{array}{ll}U U^\dagger = U^\dagger U = I \text{ and } \exists \ket{\psi} \text{ such that } U \ket{\psi} = - \ket{\psi}\\ \text{and } U \ket{\varphi} = \ket{\varphi} \text{ when } \bra{\varphi}\ket{\psi} = 0\end{array}}\;,
\end{equation*}
where the distance measure is the one induced by the (normalised) Hilbert-Schmidt norm: $d(U,V) = \sqrt{\frac1{n} \cdot \Tr\big((U - V)^\dagger (U - V)\big)}$. Note that, for every $U \in \Pi_{AK}$,
\begin{equation*}
    d(U,I) = \sqrt{\frac1{n}\sum_{i = 1}^n \sigma_i(U - I)^2} = \frac{2}{\sqrt{n}},
\end{equation*}
where $I$ is the $n \times n$ identity matrix and $\sigma_i(A)$ is the $i^\text{th}$ eigenvalue of $A$ in (say) nonincreasing order. Then $d(\Pi_{AK}, I) = 2/\sqrt{n}$, so that distinguishing between a unitary in $\Pi_{AK}$ and $I$ reduces to testing $\Pi_{AK}$ with proximity parameter $\eps \leq 2/\sqrt{n}$. Formally,
\begin{theorem}
    \label{thm:qmap-vs-qcmap}
    For any proximity parameter $\eps \leq 2/\sqrt{n}$, we have $\Pi_{AK} \in \QMAP(\eps, \log n, 1)$ and $\Pi_{AK} \notin \QCMAP(\eps, p, q)$ when $\sqrt{p} \cdot q = o(\sqrt{n})$.
\end{theorem}
\begin{proof}
    Without loss of generality, we assume query access to a \emph{controlled} unitary $U \in \Pi_{AK}$ or to the identity unitary. Note that, since the identity is $\eps$-far from $\Pi_{AK}$, distinguishing between $U \in \Pi_{AK}$ and $I$ is at least as hard as testing $\Pi_{AK}$.

    A $\QMAP$ algorithm with logarithmic proof length and query complexity 1 is as follows: given a $\log n$-qubit quantum state $\ket{\psi}$ as proof, apply the unitary to $\ket{\psi}$, accepting if and only if a phase flip is detected (by measuring and inspecting the outcome of the control qubit). Clearly, the eigenstate with eigenvalue $-1$ is a proof that causes the algorithm to accept with certainty if $U \in \Pi_{AK}$, while no quantum state is accepted if $U = I$ (also with certainty).

    The lower bound is immediate from the one proven in \cite{AK07}, which can be rephrased as follows: any QCMA algorithm that receives a proof of length $p$ and makes $q = o(\sqrt{n/p})$ oracle queries to $U \in \Pi_{AK}$ or the identity operator either accepts $I$ or rejects some element of $\Pi_{AK}$ with probability at least $1/2$.
\end{proof}

\begin{remark}
    In the usual encoding of an oracle for an $n$-bit string $x$ as a unitary $\ket{i}\ket{b} \mapsto \ket{i}\ket{b \oplus x_i}$, a Hamming distance of $\Theta(1)$ translates into $\Theta(1)$ distance in the Hilbert-Schmidt metric. Therefore, \Cref{thm:qmap-vs-qcmap} falls short of proving $\QMAP \not \subseteq \QCMAP$ (where the omitted proof and query complexities are polylogarithmic, and the proximity parameter is constant).
\end{remark}

\section{Interaction versus quantum proofs}
\label{sec:ipp-vs-qmap}

In this section we compare the power of classical interactive proofs of proximity (IPPs) and non-interactive quantum proofs of proximity (QMAPs), and show that the rather well studied problem of \emph{permutation testing} admits an efficient IPP but no efficient QMAP. In fact, for permutation testing, even an Arthur-Merlin Proof of Proximity is sufficient, as shown in \cite{GLR21}. Informally, an Arthur-Merlin Proof of Proximity (AMP) is a proof system with one round of communication where the verifier sends the first message.

Let $\Pi_P$ be the property (of \emph{functions} $f: [n] \to [n]$) defined as
\begin{equation*}
    \Pi_P = \set{ f : f \text{ is a bijection} }\;,
\end{equation*}
i.e., $\Pi_P$ is the set of all permutations. We note that in an oracle query to $f$, an algorithm sends $x \in [n]$ and receives (the $\log n$-bit string) $f(x)$; accordingly, a quantum query maps ($2\log n$-qubit states via) $\ket{x}\ket{y} \mapsto \ket{x} \ket{y+f(x)}$. Moreover, distance is measured in terms of the fraction of inputs where functions disagree (rather than with respect to their representations as bit strings), i.e., $f$ and $g$ are $\eps$-far when $\abs{\set{x \in [n] : f(x) \neq g(x)}}/n \geq \eps$.

The separation follows immediately from the two following theorems.

\begin{theorem}[{\cite[Lemma 4.2]{GLR21}}]
	\label{pi2_in_IPP}
	For every with $\eps>0$, There exists an AMP for $\eps$-testing $\Pi_P$ with query complexity $O(1/\eps)$ and communication complexity $O(\log n/\eps)$, that communicates two messages: the first is sent from verifier to prover and the second from prover to verifier. Therefore,
	\begin{equation*}
		\Pi_P \in \mathcal{AMP}(\eps, O(\log n/\eps), O(1/\eps),2)\;.
	\end{equation*}
\end{theorem}

Since $\mathcal{AMP}(\eps, c, q, r) \subseteq \IPP(\eps, c, q, r + 1)$ (by initiating the protocol with a ``dummy'' message by the prover), the following corollary is immediate.
\begin{corollary}
	$\Pi_P \in \IPP(\eps, O(\log n/\eps),O(1/\eps),3)$.
\end{corollary}

Having established that $\Pi_P$ admits an efficient IPP, we must now show it is \emph{not} efficiently testable by a QMAP:

\begin{lemma}[{\cite[Theorem 1.2]{ST19}}]
\label{pi2_not_in_QMAP}
 Any QMAP protocol for testing $\Pi_P$ with respect to proximity parameter $\eps = \Omega(1)$, using a proof of length $p$ and making $q$ queries, satisfies $p \cdot q^3 = \Omega(n)$; i.e.,
 \begin{equation*}
     \Pi_P \notin \QMAP(\eps, p,q) \text{ when } p \cdot q^3 = o(n)\;.
 \end{equation*}
\end{lemma}

Note that this implies that either $p$ or $q$ must be $\Omega(n^{1/4})$ in any QMAP protocol for $\Pi_P$. Finally, \Cref{pi2_in_IPP} and \Cref{pi2_not_in_QMAP} together imply the main result of this section.

\begin{theorem}
\label{thm:ipp-vs-qmap}
	Let $\Pi_P \subset \set{f: [n] \to [n]}$ be the set of bijective functions from $[n]$ to $[n]$, with the distance between functions $f$ and $g$ defined as $\abs{\set{x \in [n] : f(x) \neq g(x)}}/n$. Then, for any $\eps = \Omega(1)$,
	\begin{align*}
		\Pi_P &\in \IPP(\eps, O(\log n), O(1), 3) \text{ and}\\
		\Pi_P &\notin \QMAP(\eps, p,q) \text{ when } p \cdot q^3 = o(n)\;.
	\end{align*}
	In particular,
	\begin{equation*}
		\IPP(\eps, O(\log n), O(1), 3) \not \subseteq \QMAP(\eps, o(n^{1/4}),o(n^{1/4}))\;.
	\end{equation*}
\end{theorem}

\end{document}